\documentclass[format=acmsmall, review=false, screen=true]{acmart}
\usepackage{acm-ec-21}

\usepackage{booktabs} % For formal tables
\usepackage[ruled]{algorithm2e} % For algorithms

\SetAlFnt{\small}
\SetAlCapFnt{\small}
\SetAlCapNameFnt{\small}
\SetAlCapHSkip{0pt}
\IncMargin{-\parindent}

\usepackage{amsmath,amsthm,mathtools}

\usepackage{graphicx}
\usepackage{dsfont}
\usepackage{tikz}
\usetikzlibrary{calc}
\usetikzlibrary{shapes.geometric}
\usepackage{paralist}

\usepackage{wrapfig}

\usepackage{caption}
\usepackage{subcaption}

\usepackage{aliascnt}

\usepackage{xcolor}
\usepackage{hyperref}  

\hypersetup{colorlinks=true,citecolor=blue,linkcolor=blue}
\hypersetup{colorlinks=true}
\hypersetup{linkcolor=[rgb]{.7,0,0}}
\hypersetup{citecolor=[rgb]{0,.7,0}}
\hypersetup{urlcolor=[rgb]{.7,0,.7}}
\hypersetup{pdfauthor={Clayton Thomas}}

\newtheorem{theorem}{Theorem}[section]
\newtheorem*{theorem*}{Theorem}

\newaliascnt{definition}{theorem}
\newtheorem{definition}[definition]{Definition}
\aliascntresetthe{definition}

\newtheorem*{definition*}{Definition}

\newaliascnt{lemma}{theorem}
\newtheorem{lemma}[lemma]{Lemma}
\aliascntresetthe{lemma}

\newtheorem*{lemma*}{Lemma}

\newaliascnt{claim}{theorem}
\newtheorem{claim}[claim]{Claim}
\aliascntresetthe{claim}

\newtheorem*{claim*}{Claim}

\newaliascnt{fact}{theorem}

\aliascntresetthe{fact}

\newtheorem*{fact*}{Fact}

\newaliascnt{observation}{theorem}

\aliascntresetthe{observation}

\newtheorem*{observation*}{Observation}

\newaliascnt{conjecture}{theorem}

\aliascntresetthe{conjecture}

\newtheorem*{conjecture*}{Conjecture}

\newaliascnt{corollary}{theorem}
\newtheorem{corollary}[corollary]{Corollary}
\aliascntresetthe{corollary}

\newtheorem*{corollary*}{Corollary}

\newaliascnt{remark}{theorem}

\aliascntresetthe{remark}

\newtheorem*{remark*}{Remark}

\newaliascnt{proposition}{theorem}
\newtheorem{proposition}[proposition]{Proposition}
\aliascntresetthe{proposition}

\newtheorem*{proposition*}{Proposition}

\newaliascnt{example}{theorem}

\aliascntresetthe{example}

\newtheorem*{example*}{Example}
\newcommand{\notshow}[1]{}

\makeatletter
\newcommand\Ps@textstyle[2]{\mathbb{P}_{#1}\left[{#2}\right]}
\newcommand\Es@textstyle[2]{\mathbb{E}_{#1}\left[{#2}\right]}
\newcommand\Ps[2]{%
  \mathchoice % special styling in display mode, regular elsewhere.
  {\underset{{#1}}{\mathbb{P}}\left[{#2}\right]}
  {\Ps@textstyle{#1}{#2}}
  {\Ps@textstyle{#1}{#2}}
  {\Ps@textstyle{#1}{#2}}
}
\newcommand\Es[2]{%
  \mathchoice % special styling in display mode, regular elsewhere.
  {\underset{{#1}}{\mathbb{E}}\left[{#2}\right]}
  {\Es@textstyle{#1}{#2}}{\Es@textstyle{#1}{#2}}{\Es@textstyle{#1}{#2}}
}
\makeatother

\newcommand{\T}{\mathcal{T}}

\renewcommand{\Game}{\mathcal{G}}
\newcommand{\agents}{\mathcal{N}}
\newcommand{\DA}{\mathtt{DA}}
\newcommand{\successors}{\mathtt{succ}}
\newcommand{\descendants}{\mathtt{desc}}
\newcommand{\ThreeTrader}{\mathtt{3Lu}}
\newcommand{\TwoTrader}{\mathtt{2Tr}}

\title{Classification of Priorities Such That Deferred Acceptance is Obviously Strategyproof}
\author{Clayton Thomas}
\affiliation{
  \institution{Princeton University}
  \country{U.S.A.}
  }
\email{claytont@cs.princeton.edu}
\begin{abstract}
  We study the strategic simplicity of stable matching mechanisms 
  where one side has fixed preferences, termed priorities.
  Specifically, we ask which priorities are such that the strategyproofness
  of deferred acceptance (DA) can be recognized by 
  agents unable to perform contingency reasoning,
  that is, \emph{when is DA obviously strategyproof}
  (Li, 2017~\cite{Li17})?

  We answer this question by completely characterizing those priorities
  which make DA obviously strategyproof (OSP).
  This solves an open problem of Ashlagi and Gonczarowski, 2018~\cite{AshlagiG18}.
  We find that when DA is OSP, priorities are either
  acyclic (Ergin, 2002~\cite{Ergin02}), a restrictive condition which allows priorities to
  only differ on only two agents at a time, or contain an extremely limited cyclic pattern
  where all priority lists are identical except for exactly two.
  We conclude that, for stable matching mechanisms,
  the tension between understandability (in the sense of OSP) and
  expressiveness of priorities is very high.
\end{abstract}

\begin{document}

\begin{titlepage}
  \maketitle
\end{titlepage}

% \clearpage

\section{Introduction}
\label{sec:Intro}

The central task of mechanism design is achieving desirable
outcomes in the presence of strategic behaviour.
Traditionally, this is achieved through strategyproof mechanisms,
in which truth telling is a dominant strategy.
% When agents participate in a strategyproof mechanism, it is always in their
% best interest to tell the truth -- lying to the mechanism will never result
% in an outcome which the agent prefers. In other words, truth telling is a
% dominant strategy.
Suppose, however, that agents are not perfectly rational, 
and are not always able to identify their dominant strategies.
To ensure good outcomes in such an environment, the mechanism must be
strategically simple, that is, truth telling should be 
\emph{easily recognized} as a dominant strategy.

% In such environments, the social planner may want to pick their mechanism
% from a restricted class of strategyproof mechanisms for which
% strategyproofness is especially apparent to the agents.

The notion of \emph{obvious} strategyproofness (OSP) has emerged in recent
years as a fundamental criterion of strategic simplicity~\cite{Li17}.
Briefly, OSP mechanisms are those which can be recognized as strategyproof
by an agent who does not fully understand the mechanism they are
participating in, but only understands how the possible results of the game
depend on their own actions (and thus, this agent cannot perform the
contingency reasoning required to prove that certain strategies are
dominant). By eliminating the need to perform contingency reasoning, the dominant
strategies in OSP mechanisms are easier to recognize, and the game becomes
simpler for limited-sophistication agents to 
play~\cite{Li17,ZhangL17}.

% The cognitive load required to recognized dominant strategies in OSP
% mechanisms is far lower than in mechanisms which are mearly strategyproof,
% with which to measure the ability of agents
% to \emph{understand} the strategyproofness of mechanisms.

Matching environments are a crucially important environment for 
limited-sophistication agents.
% A crucially important mechanism design environment for
% limited-sophistication agents is that of one-sided matchings,
% (for instance, due to the widespread real world application of school 
% choice mechansisms).
In the real world, centralized mechanisms are used to match workers to
firms in labor markets, doctors to hospitals in residency matchings,
and students to schools in school choice environments.
In many of these markets, stability is a primary objective, and is often 
necessary to prevent market unraveling~\cite{roth2002economist}.
Formally, stability means that no pair of unmatched agents would wish to break
their match with their assigned partner, and pair with each other instead.
The celebrated deferred acceptance (DA) algorithm of~\cite{GaleS62} efficiently
determines a stable matching, and is strategyproof for one side of the
market~\cite{DubinsMachiavelliGaleShapley81}. 
These attractive theoretical properties lead many real world matching
markets to implement DA.
Unfortunately, the agents participating in these mechanisms often do not 
have the cognitive resources to precisely identify and understand
how they should participate in the mechanism,
and in practice often make costly strategic mistakes even though the
mechanism is strategyproof~\cite{HassidimMRS17,Rees-Jones18}.
This begs the question: 
why does deferred acceptance appear to be be strategically complex, despite
being strategyproof?
% why do agents have difficulty identifying their dominant strategies in
% DA, and is there a way to present DA
% in a way which makes its strategyproofness more apparent?

To get a lens on the strategic complexity of matching markets, we study
obvious strategyproofness in stable matching mechanisms with one strategic side.
In such an environment, a set of $n$ strategic
\emph{applicants} (e.g., the students being matched to schools)
is assigned in a one-to-one matching
to a set of $n$ \emph{positions} (e.g., the schools).
The applicants have \emph{preferences} over the positions,
represented as a ranked list from favorite to least favorite,
which they may misreport if they believe it will
benefit them. On the other hand, the preferences of the positions
over the applicants are fixed and non-strategic,
and thus termed \emph{priorities} over the applicants\footnote{
  No stable matching mechanism can be strategyproof for both sides of the
  market~\cite{Roth82}. Since all obviously strategyproof mechanisms are
  strategyproof, limiting our attention to one
  strategic side is necessary.
  % In the current draft of this paper, we consider the setting
  % where every positions ranks every student, that is, no matching is deemed
  % unacceptable to the positions.
  % Our results hold as written when considering the applicants submitting
  % partial preference lists.
  % Our methods may extend to a more general setting with partial
  % priority lists, and may be included in future drafts.
}. We focus on applicant-proposing 
deferred acceptance\footnote{
  DA is the unique stable matching mechanism which is
  strategyproof for one side~\cite{gale1985some, chen2016manipulability}.
  Thus, for the problem of finding obviously strategyproof stable matching
  mechanisms, restricting attention to DA proposes is without loss of generality. 
} (DA), the canonical stable matching mechanism for strategic applicants. 

Due to the crucial importance of the stable matching problem,
a refined and precise understanding of the strategic complexity of deferred
acceptance is vital.
Moreover, obvious strategyproofness has been identified as the
fundamental theoretical criterion of strategic simplicity for agents who
are unable to perform contingency reasoning.
The central question of this paper is thus: \emph{for which priorities is
deferred acceptance obviously strategyproof?}
% The most important social choice function in this setting is 
% \emph{deferred acceptance},
% which outputs the applicant-optimal stable matching, as in~\cite{GaleS62}.
% We say that a set of priorities is \emph{OSP implementable} if the one-side
% querying, applicant proposing deferred acceptance algorithm with these
% priorities is OSP implementable.
\vspace{3mm}

This paper is a follow-up to~\cite{AshlagiG18},
which was the first paper to study OSP implementations of 
stable matching mechanisms. \cite{AshlagiG18} provides a
condition on priorities which is sufficient for OSP implementability.
This condition is \emph{acyclicity}, as introduced in~\cite{Ergin02}.
Acyclicity is a strong condition on priorities, which intuitively says
that priorities may only disagree on adjacent sets of two agents.
More formally, a set of priorities is acyclic if the applicants can be
partitioned into $S_1,\ldots,S_k$, with $|S_i|\le 2$ for each $i$,
such that each position gives higher priority to applicants in $S_i$ than
applicants in $S_j$ for any pair $i,j$ with $j > i$ (see
\autoref{prop:AcyclicEquiv}).
\cite{AshlagiG18} proves that DA is OSP implementable
whenever priorities are acyclic.
However, \cite{AshlagiG18} shows that acyclicity is not necessary for
priorities to be OSP implementable.
That is, they give an example of cyclic yet OSP implementable priorities.

Unfortunately,~\cite{AshlagiG18} also proves that there exist fixed
priority sets such that DA cannot 
be implemented with an OSP mechanism.
However, they do not prove precisely which priorities are OSP
implementable, and thus leave open the possibility that there are
OSP implementable priorities which are ``diverse'' in some useful sense
(that is, there may be priorities which allow for enough variation between the positions
to capture practically useful constraints in some matching environment).

% Moreover, \cite{AshlagiG18} argue that, for most randomly drawn
% sets of priorities, deferred acceptance will not be OSP implementable.

% However, due to the lack of completeness of their classification, there was
% still a hope for diverse priorities which could be OSP implemented.

% all \emph{acyclic} priorities define
% OSP mechanisms, but that there exist cyclic priorities which are also
% OSP.

We show that for stable matching mechanisms,
the gap between OSP implementable priorities and acyclic priorities
is very small.
That is, by classifying which priorities are OSP implementable,
we show that very little is possible with cyclic 
but OSP implementable priorities.
\vspace{3mm}

% \ctnote{organize?}
% \begin{theorem*}[Informal, 
%     see \autoref{thrm:Construction} and \autoref{thrm:MainImposibility}]
%   Priorities are OSP implementable if and only if they are acyclic, or the
%   priorities over each ``set of cyclic applicants'' is of a specific form in which
%   all positions except for exactly two have the same priority over this set
%   of applicants.
% \end{theorem*}

We now give one example of cyclic but OSP priorities for $6$ applicants
and $6$ positions. We write priorities as a list,
starting from the top priority and proceeding downwards.
In this example,
positions $1,2,3,4$ have the same priority over all applicants,
but $5$ and $6$ have slightly different priorities (we highlight the
pairs of applicants at which $5$ and $6$ differ from $1,2,3,4$ using
parentheses):
%   \begin{tabular}{ccccccc}
%     1,2,3& :\ $a$ $\succ b$ $\succ c$ $\succ d$ $\succ e$ \\
%     4    & :\ $a$ $\succ c$ $\succ b$ $\succ e$ $\succ d$ \\
%     5    & :\ $b$ $\succ a$ $\succ d$ $\succ c$ $\succ e$
%   \end{tabular}
% \\
% \begin{tabular}{c|ccccccccccc}
%   1,2,3,4 & $a$ &  $b$  &  $c$  &  $d$  & $e$  & $f$ \\
%   5     & $a$ &  $(c$ &  $b)$ &  $(e$ & $d)$ & $f$ \\
%   6     & $(b$ &  $a)$ &  $(d$ &  $c)$ & $(f$ & $e)$
% \end{tabular}
\begin{align*}
  1,2,3,4: &\  a\ \succ\ b \succ\ c\ \succ\ d \succ\ e\ \succ f  \\
  5:       &\  a\ \succ (c \succ b)\succ (e \succ d)\ \succ f 
  \tag*{(*)}\label{eqn:star} \\
  6:       &\ (b \succ a)\succ (d \succ c) \succ (f \succ e)
\end{align*}
It turns out that all cyclic but OSP priorities can be constructed from
acyclic priorities, and priorities that are exactly like this example
(generalized to any number of applicants).
In particular, when priorities are cyclic on some set of applicants, 
every position except for precisely two must have the same priority list
(and moreover, the remaining two must differ in a precise pattern which
flips only adjacent applicants).
We term such priorities \emph{limited cyclic}
(\autoref{def:limitedCyclic}).
% We call these set of priorities \emph{limited cyclic} priorities.

When priorities are acyclic, \cite{AshlagiG18} shows that
an OSP mechanism for DA constructed from simple
combinations of a mechanism which we call $\TwoTrader$.
Their mechanism proceeds in a way somewhat similar to 
a serial dictatorship (in which applicants arrive in some fixed order 
and are permanently matched to
their favorite remaining position) but where two applicants 
may arrive simultaneously.
% may be
% ``active'' at a given time (that is, at certain stages of the mechanism
% there are agents who have acted, but who's match is not yet determined).
These two agents are assigned to the remaining positions via the
mechanism $\TwoTrader$.
When priorities are limited cyclic, it turns out that an OSP mechanism can
be constructed from a combination\footnote{
  In the combined mechanism, no agent ever participates in two distinct
  ``sub-mechanisms'' $\TwoTrader$ or $\ThreeTrader$.
  This gives a strong sense in which ever OSP stable matching mechanism can
  be built out of $\TwoTrader$ and $\ThreeTrader$.
} of $\TwoTrader$ and a fairly simple 
mechanism we term $\ThreeTrader$,
which interacts with three applicants at a time.
% which interacts with three applicants
% in which three applicants are active at a certain time.
Interestingly, this implies that in some sense there are only 
two ``irreducible'' obviously strategyproof stable matching mechanisms, 
$\TwoTrader$ and $\ThreeTrader$.
% \ctnote{revise}
% (following the terminology of \cite{Troyan19}).
% The two ``active'' agents are offered disjoint sets of positions which they
% may choose to be immediately matched with.
% If neither 
% in an OSP manner without either applicant being able to guarantee
% themselves their favorite remaining position.

The crux of our characterization revolves around proving that when
priorities are not limited cyclic, DA is not OSP implementable.
We sketch how this proof proceeds below.
Interestingly, the structure of the proof admits another
characterization of OSP implementable priorities. 
Specifically, a set of priorities is OSP implementable if and only if no
restriction of the priorities (i.e. restricting attention to some subset of
positions and applicants) is equal to one of the priority sets
listed in \autoref{fig:AllPrimitiveNonOsp}.
Thus, our characterization not only implies that there are essentially only
two OSP stable matching mechanisms, but that there is in some sense 
a finite list of ``irreducible'' non-OSP priority sets.
% is so restrictive 

% that even for three applicants and three positions, there is only one
% cyclic yet OSP implementable set of priorities (up to relabeling).
% Because 
% 
% in such a way that there are very few
% cases to check for four applicants and positions (and no cases to check for
% more than four agents, simply by pattern matching) \ctnote{todo: actually
% write lol}.

All told, our characterizations can be concisely states as follows.
We prove our main theorem in \autoref{sec:Proof} as a simple 
combination of \autoref{thrm:Construction},
\autoref{thrm:figSetsNotOsp}, and \autoref{thrm:patternMatchingMainThrm}.
\begin{theorem*}[Main Theorem]
  For any set of priorities $q$,
  the following are equivalent:
  \begin{enumerate}
    \item Deferred acceptance with priorities $q$ is OSP implementable.
    \item $q$ is limited cyclic.
    \item An OSP mechanism for deferred acceptance with priorities 
      $q$ can be constructed from compositions of
      $\TwoTrader$ and $\ThreeTrader$.
    \item No restriction of $q$ is equal, up to relabeling, to any of the
      priority sets exhibited in \autoref{fig:AllPrimitiveNonOsp}.
  \end{enumerate}
\end{theorem*}

% \begin{theorem*}[Informal, 
%     see \autoref{thrm:Construction} and \autoref{thrm:MainImposibility} 
%   \ctnote{add or fix?} ]
% \end{theorem*}

\begin{figure}
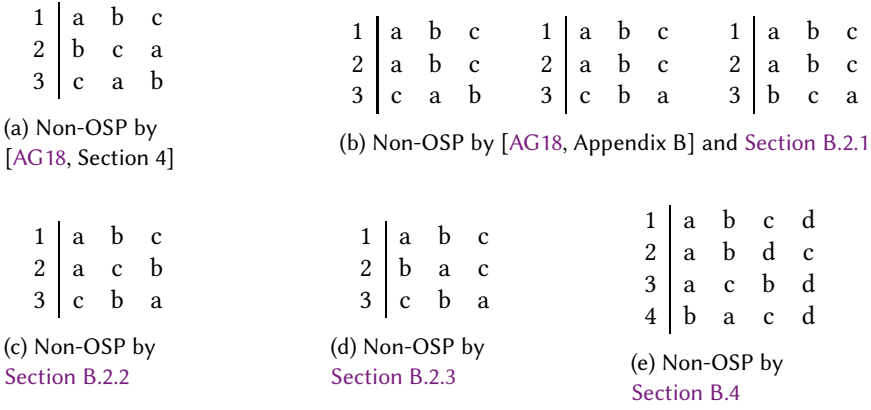

  % \begin{subfigure}{\textwidth}
  %   \begin{tabular}{c|ccc}
  %     1 & a& b& c \\
  %     2 & a& b& c \\
  %     3 & a& b& c
  %   \end{tabular}
  %   \qquad
  %   \begin{tabular}{c|ccc}
  %     1 & a& b& c \\
  %     2 & a& b& c \\
  %     3 & b& a& c
  %   \end{tabular}
  %   \qquad
  %   \begin{tabular}{c|ccc}
  %     1 & a& b& c \\
  %     2 & a& b& c \\
  %     3 & a& c& b
  %   \end{tabular}
  %   \caption{OSP by acyclicity}
  % \end{subfigure}

  % \begin{subfigure}{\textwidth}
  %   \begin{tabular}{c|ccc}
  %     1 & a& b& c \\
  %     2 & b& a& c \\
  %     3 & a& c& b
  %   \end{tabular}
  %   \caption{OSP but cyclic case}
  % \end{subfigure}

  \begin{tabular}{cccc}

  \begin{subfigure}{25mm}
    \centering
    \begin{tabular}{c|ccc}
      1 & a& b& c \\
      2 & b& c& a \\
      3 & c& a& b
    \end{tabular}
    \caption{Non-OSP by \\ \cite[Section 4]{AshlagiG18}}
  \end{subfigure}
   &
   ~\qquad\ %
   &

  \multicolumn{2}{c}{
  \begin{subfigure}{80mm}
    \centering
    \begin{tabular}{c|ccc}
      1 & a& b& c \\
      2 & a& b& c \\
      3 & c& a& b
    \end{tabular}
    \quad
    \begin{tabular}{c|ccc}
      1 & a& b& c \\
      2 & a& b& c \\
      3 & c& b& a
    \end{tabular}
    \quad
    \begin{tabular}{c|ccc}
      1 & a& b& c \\
      2 & a& b& c \\
      3 & b& c& a
    \end{tabular}
    \caption{Non-OSP by~\cite[Appendix B]{AshlagiG18} and~\autoref{sec:NonOSPTwoWomenSame}}
  \end{subfigure}
  }
  \\
  \\

  \begin{subfigure}{25mm}
    \centering
    \begin{tabular}{c|ccc}
      1 & a& b& c \\
      2 & a& c& b \\
      3 & c& b& a
    \end{tabular}
    \caption{Non-OSP by\\\autoref{sec:NonOspMildExtension}}
  \end{subfigure}
  &
  \quad
  &

  \quad
  \begin{subfigure}{25mm}
    \centering
    \begin{tabular}{c|ccc}
      1 & a& b& c \\
      2 & b& a& c \\
      3 & c& b& a
    \end{tabular}
    \caption{Non-OSP by\\\autoref{sec:NonOspLastCase}}
  \end{subfigure}
  &

  \begin{subfigure}{25mm}
    \centering
    \begin{tabular}{c|cccc}
      1 & a& b& c & d \\
      2 & a& b& d & c \\
      3 & a& c& b & d \\
      4 & b& a& c & d
    \end{tabular}
    \caption{Non-OSP by\\\autoref{sec:nIs4CasesProof}}
  \end{subfigure}

  \end{tabular}

  \caption{ Each set of priorities displayed here is not OSP implementable.
    Moreover, any set of priorities which is not OSP implementable
    contains one of these sets of priorities as a sub-pattern (up to relabeling).
    }
  \label{fig:AllPrimitiveNonOsp}

\end{figure}

\subsection{Intuition and layout of the proof}
\label{sec:intuition}

We now give informal intuition as to why one might expect obvious
strategyproofness to place such severe restriction on priority sets.
Obvious strategyproofness requires that every time an agent acts the
mechanism, the worst thing that can happen if they play truthfully (for the
remainder of the mechanism) is no worse than the best thing that 
can happen if they lie.
When priorities are acyclic, there are at most two applicants 
% (call them $a,b$)
who are ranked first by some set of position.
As shown in~\cite{AshlagiG18}, these applicants can be assigned to positions
in an OSP manner via the mechanism $\TwoTrader$,
which we recall in \autoref{fig:IllustrationAcyclic}.
The mechanism $\TwoTrader$ communicates with only these top two
applicants and queries these applicants at most twice,
that is, only those two applicants are ``active''\footnote{
  Formally, at some state $h$ of the mechanism, an applicant $a$ is
  \emph{active} if either $a$ moves at $h$, or if $a$ has moved in the
  past, but their match is not yet completely determined.
}, and each applicant ``moves'' at most twice.
% Briefly, $\TwoTrader$ first asks $a$
% ``is your favorite position some $x$ where you have top priority''?
% If yes, $a$ is matched to $x$, if not, $\TwoTrader$ asks $b$ to pick and
% match to their favorite position, then asks $a$ to pick and match to their
% favorite remaining position. This correctly assigns $a$ and $b$ as in
% deferred acceptance.
% The key to seeing that this is OSP is that 
In order for \emph{cyclic} priorities to be OSP,
one of two things must happen in the mechanism:
some applicant must move more than twice, or
more than two applicants must be active at some point.

\begin{figure}
  \begin{tabular}{cc}
    \begin{minipage}[b]{1.3in}
      % \vspace{-1.2in}
      \begin{align*}
        U: 
        \begin{cases}
          \ a \succ b \succ \ldots \\
          \ a \succ b \succ \ldots \\
          \qquad \vdots
        \end{cases}
        \\
        V: 
        \begin{cases}
          \ b \succ a \succ \ldots \\
          \ b \succ a \succ \ldots \\
          \qquad \vdots
        \end{cases}
      \end{align*}
    \end{minipage}
    % \vspace{-0.5in}
  & 

\tikzset{
  solid node/.style={circle,draw,inner sep=1.5,fill=black},
  decision node/.style={regular polygon, regular polygon sides=6,draw,inner sep=1.5},
  hollow node/.style={circle,draw,inner sep=1.5},
  square node/.style={draw,inner sep=2.5},
  % node gap={xshift=70} 
}
% \resizebox{3.5in}{2in}{
  \begin{tikzpicture}[scale=0.7,font=\footnotesize]
    \tikzstyle{level 1}=[level distance=20mm,sibling distance=35mm]
    \tikzstyle{level 2}=[level distance=15mm,sibling distance=15mm]
    \node(a1)[decision node,label=above:{}]{(a)}
      child{node(b-sub)[decision node]{(b)}
        child{node(a11)[square node, label=below:{$\vdots$}]{}
          edge from parent node[left]{
            \begin{tabular}{l}
              $\mathtt{Clinch}$ any \\ $v \in U\cup V\setminus \{u\}$
            \end{tabular}
          }
        }
        edge from parent node[left]{
          \begin{tabular}{l}
            $\mathtt{Clinch}$ \\ any $u \in U$
          \end{tabular}
        }
      };
    \node(b1)[decision node,label=above:{}, xshift=150]{(b)}
      child{node(b13)[decision node]{(a)}
        child{node(a-sub)[square node, label=below:{$\vdots$}]{}
          edge from parent node[left]{
            \begin{tabular}{l}
              $\mathtt{Clinch}$ any \\ $u \in U\cup V\setminus \{v\}$
            \end{tabular}
          }
        }
        edge from parent node[left]{\begin{tabular}{l}
            $\mathtt{Clinch}$ \\ any $v \in U \cup V$
          \end{tabular}
        }
      };
    \draw[->] (a1) -- (b1) 
      node[midway, above]{$V\ldots$};
  \end{tikzpicture}
  % \vspace{-0.5in}
% }

  \end{tabular}

  \caption{
    An example of acyclic priorities, and the mechanism $\TwoTrader$
    which matches the ``top two'' applicants $a, b$ (who have top priority
    at the set of positions $U$ and $V$, respectively).
    ``\texttt{Clinch}ing'' a position $x$ allows an applicant to exit the
    mechanism and be perminently matched to that position.
    At the top right node where applicant $b$ acts, applicant $a$ has
    moved, but their match is not yet determined, so two applicants are
    active.
    This mechanism is OSP because at the second node where $a$ acts (i.e.
    the bottom right node), $a$ can either clinch their favorite position
    in $V$, or (if $b$ already clinched that position) they can clinch
    anything in $U$ (the set they were offered to clinch at the first node
    where they acted).
    }
  \label{fig:IllustrationAcyclic}

\end{figure}
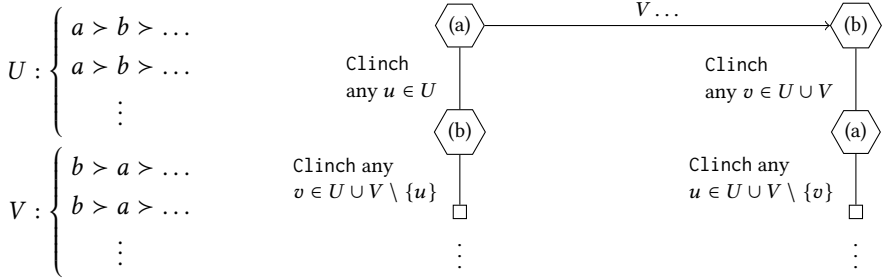

Intuitively, no applicant can make more than two non-trivial moves in an
OSP implementation of deferred acceptance because 
\emph{it is hard for an applicant to ever
guarantee themselves a position where they do not have top priority}.
This is because when some applicant $a$ proposes to such a position $x$,
there is almost always a risk that an applicant with higher 
priority proposes to that position later.
If an OSP mechanism learns from applicant $a$ that their favorite
position is $x$, but $a$ is later rejected from $x$,
they must be able to guarantee themselves every position which was possible
at the first time they acted.
Thus, the mechanism cannot again ask applicant $a$ further questions which
result in uncertainty for $a$, i.e. the second question asked to $a$ should
determine $a$'s final match.

It is possible for three applicants to be active at some point in an OSP 
implementation of deferred acceptance. Indeed, this is the case in
$\ThreeTrader$. However, this can happen only in a very specific way.
% Now, suppose three applicants are active at some point in the mechanism.
% Call these applicants $a$, $b$, and $c$ (as in \autoref{eqn:star}).
% Consider the point in deferred acceptance when the first (say, $a$)
% of the set of three active agents is chosen to participate.
Let $a$ be the first of those three applicants to move.
It turns out that in order for this instance of DA
to be OSP, the mechanism must deduce from the first question asked to
$a$ that one of the other applicants (call them $c$)
cannot possibly achieve some position.
Due to the definition of DA, this means the mechanism must
know $a$'s favorite position, and $c$ must have a lower
priority than $a$ at this position.
But for a matching mechanism to be OSP, the only way it can deduce that
$a$'s favorite position is $x$ is to offer $a$ to ``clinch'' 
every position other than $x$\footnote{
  For the reader familiar with~\cite{BadeG16}, this is precisely the
  condition under which applicant $u$ may become a ``lurker'' for position
  $x$.
}, that is, $a$ must be able to guarantee themselves every position other
than $x$ if they choose.
\begin{wrapfigure}{R}{24mm}
    \begin{tabular}{c|ccc}
      1 & a& b& c \\
      2 & a& c& b \\
      3 & b& a& c
    \end{tabular}
    \caption{OSP but cyclic case}

  \label{fig:ThreeCyclicOsp}
\end{wrapfigure}
Thus, $a$ must have top priority at every position other than $x$.
For $a$ to remain active after they move, there must be an applicant $b$
such that $x$ has priority $b \succ a \succ c$.
Moreover, for $b$ and $c$ to be active at the same time, there must also be a
position $y$ for which $c \succ b$.

This informal argument begins to explain why the priority set of
\autoref{fig:ThreeCyclicOsp} might be the only cyclic but OSP priority set
for three applicants and three positions.
This is indeed the case. Moreover, removing applicant $a$ from
consideration and applying the argument recursively, we'd expect that $b$
should be the top applicant (other than $a$) at every position except for
one. Indeed, limited cyclic priority sets such as \autoref{eqn:star}
satisfy a recursive property like this, in which removing their top
applicants yields another limited cyclic priority set of the same pattern.
This gives intuition for the structure of the proof as well -- the highly
demanding constraints which OSP imposes apply recursively,
and generalizations of \autoref{fig:ThreeCyclicOsp}
such as \autoref{eqn:star} are the only examples of cyclic but
OSP priority sets.

% But the same informal argument as above hints that there should only be
% \emph{one} such position $y$.
% In \autoref{eqn:star}, position $x$ is $6$ and $y$ is $5$.
% Removing applicant $a$ and applying this logic recursively
% indicates that in \autoref{eqn:star}, there should be a unique applicant
% $d$ which , 
% and it turns out it can only be y (bleh).

More technically, our proof proceeds by first enumerating all cyclic
priorities for three applicants and three positions.
It turns out that up to relabeling there is only one such case which is
OSP (namely, that of \autoref{fig:ThreeCyclicOsp}),
and all the others (namely, (a)-(d) of \autoref{fig:AllPrimitiveNonOsp})
are non-OSP.
As pointed out in \cite{AshlagiG18}, a set of priorities can be OSP only if
every restriction of those priorities (to some subset of applicants and
positions) is also OSP (we recall this in \autoref{lem:RestrictedPriorities}).
Thus, the fact that only one cyclic case for three applicants and positions 
is OSP gives a very controlled form which cyclic and OSP priorities of four
applicants and positions could possibly take.
The only such nontrivial cases are limited cyclic cases, which are OSP,
and (e)~from \autoref{fig:AllPrimitiveNonOsp}, which is not OSP.
This again gives a very controlled form which cyclic and OSP priorities of
more than four applicants and positions can take,
and indeed suffices to show that all OSP priorities must be limited cyclic 
for any number of agents.

% At the start of the mechanism, 
% these applicants can ``clinch'' any of the positions where they are ranked
% first (that is, they may choose to be immediately matched to these
% positions and leave the mechanism).
% If the first such applicant $u$ does not choose to clinch any of these positions, 
% that means $u$'s favorite position lies in set $V$.
% Two things may happen when : the second applicant may cl
% In an OSP mechanism,

\subsection{Discussion}
While our result allows for a general class of cyclic but OSP priorities, 
priorities like \autoref{eqn:star} are arguably
\emph{less} ``diverse'' than acyclic preferences.
In acyclic preferences, half of all positions may give some applicant $a$
in their top priority, but the other half may give some applicant $b$
top priority.
But our characterization shows that when priorities over some applicants
are cyclic, the priorities of all positions must be
identical except for exactly two of them.
Thus, the expressiveness of OSP priorities is very limited --
priorities may either differ on only two applicants at a time,
or may differ at only two positions (in a very controlled way).

% It is instructive to compare our result to a conceptually similar
% characterization provided by~\cite{Troyan19} for the \emph{top trading
% cycles} mechanism.
% Top trading cycles is a mechanism which, similarly to deferred acceptance,
% is used to create a matching on the basis of priorities. 
% While deferred acceptance focuses on stability, top trading cycles focuses
% on creating a Pareto optimal matching.
% In characterizing those priority sets for which top trading cycles is OSP,
% \cite{Troyan19} determines that 

It is instructive to compare our characterization of OSP deferred
acceptance mechanisms to the OSP implementations of \emph{top trading
cycles} presented in~\cite{Troyan19}.
OSP implementations of top trading cycles are also fairly restrictive, 
but may, for example, gradually endow an applicant with a larger and larger
set of positions they might ``clinch'', querying them many times.
If the applicant's favorite position is not in this ``endowed'' set, the
applicant will will not clinch any position, and based on the actions of
other applicants, they may be gradually endowed with more positions and be
called to act in the mechanism many times.

In contrast, in the mechanisms $\TwoTrader$ and $\ThreeTrader$,
and thus in all OSP implementations of deferred acceptance, applicants
are called to act at most two times.
This may indicate that deferred acceptance is fundamentally more
``strategically complex'' than top trading cycles,
because the obvious strategyproofness constraint seems more demanding for
deferred acceptance than for top trading cycles.
That is, the ``more complex'' deferred acceptance mechanism should have a
smaller OSP subset, and its possible OSP interactions with the applicants
should be more limited\footnote{
  This conforms to our intuition for why OSP implementation of DA
  cannot ask agents to move more than twice. In DA, it is hard to guarantee
  an applicant any position where they do not have top priority. On the
  other hand, in the middle of the execution of top trading cycles,
  an applicant is guarantee to
  receive a position (at least as good as) any position which currently
  points to them (directly or indirectly).
  OSP implementations of top trading cycles harness exactly this fact in
  order to (potentially) interact with agents many times.
}.
% This conforms to the intuition that the strategyproofness of top trading
% cycles is easier to recognize that
On the other hand, it is possible for three agents to be active at a time
in OSP implementations of deferred acceptance (during the mechanism
$\ThreeTrader$), which can never happen in top trading cycles,
so the comparison is nuanced.

Limited cyclic priorities, and thus all cyclic but OSP implementable
priorities, are a generalization of the example used
by~\cite{AshlagiG18} to demonstrate that acyclicity is not necessary for
OSP implementability (and thus our mechanism $\ThreeTrader$ generalizes the
mechanism which they constructed).
This came as quite a surprise to us.
Indeed, our first goal in pursuing this question was to identify
interesting new cases of deferred acceptance with fixed priorities 
(for example, ones where more than $3$ applicant are active at a time,
or where an applicant moves more than twice).
In the end, we proved that there are no such cases.
We believe that this result is a further demonstration of the opinion that,
while obvious strategyproofness is a fundamental notion of simplicity in
mechanism design, new definitions are needed to allow for practical
mechanisms which are still strategically simple or easy to explain.
Studying the ``irreducible'' non-OSP priorities of
\autoref{fig:AllPrimitiveNonOsp} may be a small stepping stone to
formulating such a definition.

\subsection{Additional related work}

There is a rapidly growing body of work on obvious
strategyproofness and various related notions of strategic simplicity.
\cite{Li17} introduced OSP, used this notion to study ascending price auctions,
and proved that serial dictatorship is OSP implementable but
top trading cycles with initial endowments is not OSP
implementable. \cite{AshlagiG18}, the most related work to the present
paper, was the first released follow-up work to \cite{Li17}.
\cite{Troyan19} and \cite{mandal2020obviously} study top trading cycles 
without initial endowments, and
classify those priorities which make top trading cycles OSP implementable. 
Our classification completes the line of study initiated 
by~\cite{AshlagiG18}, answering their main open questions.
The present paper is significantly more involved than prior works
on OSP matching mechanisms on the basis of priorities,
and to complete our classification we
must carefully reason about sub-patterns of priority sets as well as prove
that several specific matching mechanisms induced by certain 
priority sets are not OSP implementable.

\cite{BadeG16} studies a range of social-choice settings, and investigates
the possibility of OSP and Pareto optimal social choice functions.
Deferred acceptance with cyclic preferences is not Pareto optimal
\cite{Ergin02}, and thus falls outside this paradigm.
\cite{PyciaT19} provides various simplification results for a broad class
of domains which they call ``rich''.
\cite{Bade19} constructs an interesting OSP mechanism for the
house-matching problem with initial endowments in which agents have
single-peaked preferences.
\cite{Mackenzie20} studies OSP mechanisms in great generality, proving
several equivalences and simplifications.
\cite{ZhangL17} provides a decision-theoretic framework justifying obvious
dominance.
\cite{golowich2021computational} provides a general condition over any
environment which is equivalent to OSP implementability, and provides an
algorithm for checking this condition in time which is polynomial in the
size of a table for computing the social choice function (i.e. a table 
listing the result for each possible tuple of agent's types).

Refinements and generalizations of the notion of OSP have been made.
\cite{PyciaT19} introduce the notions of strongly obviously strategyproof
and one-step foresight obviously strategyproof mechanisms, which are a
subset of OSP mechanisms which are even ``simpler'' to recognize as
strategyproof.
\cite{TroyanM20} introduces the notion of non-obviously manipulable
mechanisms, which generalize \emph{strategyproof} mechanisms to allow for
some strategic manipulations (but not ``obvious'' ones).
\cite{LiD20} goes in another distinct direction, and designs
revenue-maximizing mechanisms for strategically simple agents
(who may, for example, play dominated strategies but will not play
obviously dominated strategies). 

\paragraph{Acknowledgements.}
We thank Yannai Gonczarowski for invaluable guidance in the presentation of
this manuscript.
We thank
Itai Ashlagi,
Linda Cai,
and Matt Weinberg
for helpful discussions.

\section{Preliminaries and prior results}
\label{sec:Prelims}

% The major points:
% 
% \begin{itemize}
%   \item if a subpattern is non-OSP then the whole thing is non-OSP.
% \end{itemize}

\subsection{Mechanisms and Obvious Strategyproofness}
\label{sec:MechDef}

% An \emph{(ordinal) matching environment} $E = (\agents, Y, (\T_i)_{i\in \agents})$ 
% consists of a set of applicants $\agents$, positions $Y$,
% and preferences (also called \emph{types}) $\T_i$ for each $i\in\agents$.
% % a set of agents (also known as players) 
% % $\agents = [n]$, and types $\T_i$ for each agent $i\in \agents$.
% Each type $\succ_i\in\T_i$ corresponds is a strict total order
% over positions in $Y$. 
% which represents how agent $i$ ranks each outcome in $Y$ when their type
% is $t_i$.
% Let $Z$ denote the set of all matchings 
% A \emph{social choice function} over $E$ is a mapping
% $f : \T_1\times\dots\times\T_n \to Z$ from types of each agent to a.
% \footnote{
%   A total order $\succeq$ over $Y$ is a reflexive, 
%   transitive, total, and antisymmetric binary relation on $Y$.
%   A strict total order $\succ$ 
%   % if $a \succeq b, b\succeq c$, then $a \succeq c$,
%   for all $x,y\in Y$, we have one of
%   $x \succeq y$ or $y \succeq x$.
%   We write $a \succ b$ when we have $a\succeq b$ but we do not have
%   $b \succeq a$.
% },

An \emph{(ordinal) environment} $E = (Y, (\T_i)_{i\in \agents})$ 
consists of a set of outcomes $Y$, a set of agents (also known as players) 
$\agents = [n]$, and types $\T_i$ for each agent $i\in \agents$.
Each type $t_i\in\T_i$ corresponds to a weak order $\succeq_i^{t_i}$
over outcomes in $Y$\footnote{
  A weak order over $Y$ is a reflexive and transitive binary relation
  on $Y$ such that,
  % if $a \succeq b, b\succeq c$, then $a \succeq c$,
  for all $x,y\in Y$, we have one of
  $x \succeq y$ or $y \succeq x$.
  We write $a \succ b$ when we have $a\succeq b$ but we do not have
  $b \succeq a$.
}, which represents how agent $i$ ranks each outcome in $Y$ when their type
is $t_i$.
A \emph{social choice function} over $E$ is a mapping
$f : \T_1\times\dots\times\T_n \to Y$ from types of each agent to outcomes.
The most commonly studied notion in mechanism design is that of a
strategyproof social choice function:
% This definition is a property only of the way types
% are mapped to outcomes.
\begin{definition}
  Social choice function $f$ is \emph{strategyproof} if,
  for any $i\in\agents$, $t_1\in\T_1,\ldots,t_n\in\T_n$, and
  $t_i'\in\T_i$, we have
  \[ f(t_i, t_{-i}) \succeq_i^{t_i} f(t_i', t_{-i}). \]
\end{definition}

To study cognitively limited agents (who may not fully understand the game
they are playing), one needs to also know how the social choice function is
implemented, i.e. how the social planner interacts with the strategic
agents to compute $f$. This is done via a mechanism.
A \emph{mechanism}\footnote{
  In technical terms, we consider a mechanism to be a deterministic
  perfect information extensive form game with consequences in $Y$.
  % For OSP implementations, it is known 
  As in \cite{AshlagiG18}, restricting attention to these mechanisms
  is without loss of generality.
  % For completeness, we provide a detailed discussion of definitions
  % in \autoref{sec:DefinitionDiscussion}.
} over $E$ is a tuple
$\Game = (H, E, \mathtt{Pl}, (\T_i(\cdot))_{i\in \agents}, g)$ such that:
\begin{itemize}
  \item $H$ is a set of states (also called nodes),
    and $E$ is a set of directed edges
    between the states, such that $(H, E)$ forms a finite directed tree
    (where every edge points away from the root).
    Let the set of leaves be denoted $Z$, 
    and let the root be denoted $h_0$.
    For any $h\in H$, let $\successors(h)\subseteq H$
    be the set of nodes which are immediate successors of $h$ in the game
    tree.
    % , and let $\descendants(h) \subseteq H$ be those nodes which are
    % descendants of $h$ in the game tree.

    % We write edges like $(h,h')\in E$, where $h$ is between the root and $h'$.
    % (i.e. those edges leading to successor states, so that $(h, h')\in
    % \sigma(h)$ when $h'$ is a state immediately following $h$).
  \item $\mathtt{Pl} : H \setminus Z \to \agents$ is the player
    choice function, which labels each non-leaf node in $H\setminus Z$ with
    the agent who acts at that node.
    % We assume without loss of generality that no agent makes consecutive
    % moves, i.e. for all $h$ and $h'\in \successors(h)$, we have
    % $\mathtt{Pl}(h)\ne\mathtt{Pl}(h')$.
  \item For each $i\in\agents$ and $h\in H$,
    we have $\T_i(h)\subseteq \T_i$ the \emph{type-set of agent $i$ at $h$}.
    These must satisfy
    \begin{itemize}
      \item % When player $i$ acts at $h$, their actions partition $\T_i(h)$.
        % That is, 
        For any $h$ with $\mathtt{Pl}(h)=i$, the sets
        $\{ \T_i(h') \}_{h'\in\successors(h)}$ form a partition of $\T_i(h)$.
      \item If $\mathtt{Pl}(h)\ne i$, then for all $h'\in\successors(h)$,
        we have $\T_i(h')=\T_i(h)$.
    \end{itemize}

    % Moreover, we assume without loss of generality that each $\T_i(h)\ne
    % \emptyset$ (so that all nodes are reached by some type),
    % and that $|\successors(h)|\ge 2$ for all $h$
    % (so that all actions are nontrivial).

    % Define $H_i := \{ s \in S | \mathcal{P}(s)=i \}$ as the set of states
    % where player $i$ is called to act.
    % We assume that no player takes a consecutive turn,
    % that is, for any $h$ and $h'\in \sigma_H(h)$, we have
    % $\mathcal{P}(h)\ne\mathcal{P}(h')$\footnote{
    %   Note that this assumption is without loss of generality.
    %   Condensing consecutive actions by the same player into a single node
    %   leaves the game essentially unchanged.
    %   In particular, the communication cost does not increase and if the 
    %   mechanism was already DSIC, then it is still DSIC.
    % }.
    % where the players are $1, \ldots, n$, and
    % $0$ represents a ``chance node''.
  \item $g : Z \to Y$ labels each leaf node with an outcome in $Y$.
\end{itemize}

For types $t_1\in\T_1,\ldots,t_n\in\T_n$, we let
$\Game(t_1,\ldots,t_n)$ denote $g(\ell)$, where
$\ell$ is the unique leaf node of $H$ with 
$(t_1,\ldots,t_n)\in\T_1(\ell)\times\ldots\times\T_n(\ell)$.
We say $\Game$ \emph{implements} social choice function $f$
when $\Game(t_1,\ldots,t_n) = f(t_1,\ldots,t_n)$.
Note that the collection of nodes $h$ where
$(t_1,\ldots,t_n)\in\T_1(h)\times\dots\times\T_n(h)$ always forms a path
from the root to the leaf $\ell$. 
Intuitively, this path is the ``execution path'' of $\Game(t_1,\ldots,t_n)$.
A node $h$ where $\mathtt{Pl}(h)=i$ represents the
mechanism asking player $i$ the question ``which of the sets
$\{ \T_i(h') \}_{h'\in \successors(h)}$ does your type lie in?''.
If agent $i$ with type $t_i$ plays truthfully, 
the mechanism only reaches node $h$ if $t_i \in \T_i(h)$,
and the agent will ``play the action'' leading to the unique $h'$
with $t_i\in\T_i(h')$.

Obvious strategyproofness (OSP) can then be defined as follows:

\begin{definition}[OSP, \cite{Li17, AshlagiG18}]
  \label{def:OSP}
  Mechanism $\Game$ over $E$ is \emph{OSP at node $h\in H$}
  if the following holds:
  If $\mathtt{Pl}(h)=i$, for every $t_i \in \T_i(h)$,
  let $h'\in\successors(h)$ be the unique successor node of $h$ such that
  $t_i\in\T_i(h')$. 
  Then for every leaf $\ell$ which is a descendant of $h$ such that
  $t_i \in \T_i(\ell)$,
  and every leaf $\ell'$ which is a descendant of $h$ but not $h'$,
  % and every $\ell'$ where $\ell'$ is a descendent
  % of a node other than the action taken by $t_i$
  % at $h$, we have
  we have $ g(\ell) \succeq_i^{t_i} g(\ell') $.

  Mechanism $\Game$ is \emph{OSP}
  if it is OSP at every node $h\in H$.

  Social choice function $f$ over $E$ is \emph{OSP implementable}
  (simply called OSP for short) if there exists an OSP mechanism which
  implements $f$.
\end{definition}

In words, a mechanism is OSP if, for every node $h$ where an agent $i$ is
called to act, the \emph{worst} thing that can happen if $i$ continues to
play truthfully for the rest of the game 
is no worse for $i$ than the \emph{best} thing that could
possibly happen if $i$ deviates at node $h$.
% \footnote{
%   Note that in an OSP mechanism,
%   it is possible that at $h$ there are two leaves 
%   $\ell, \ell' \in \descendants(h)$, where
%   $t_i \in \T_i(\ell)$ but $t_i\notin \T_i(\ell')$,
%   and where $g(\ell') \succ_i^{t_i} g(\ell)$.
%   That is, it is possible that the best case under non-truthful play is
%   preferred to the worst case under truthful play.
%   However, by \autoref{def:OSP}, such an $\ell'$ cannot be
%   in $\descendants(h) \setminus \descendants(h')$, where 
%   $h'\in \successors(h)$ is the node corresponding to the action which agent
%   $i$ takes under truthful play (that is, $\ell'$ must be in
%   $\descendants(h')$).
%   In words, agent $i$ cannot reach leaf $\ell'$ by taking a non-truthful
%   action at $h$.
%   % \ctnote{Perhaps I should not use the phrase ``In words,'' three times in a row.}
%   Intuitively, truthful play always seem like a good idea in the moment, 
%   even if the agents only understand how the outcome depends on their own
%   actions (and the agent does not understand how the game depends on the
%   actions of other players). 
% }.
In particular, this implies that the mechanism (and the social choice
function it implements) is strategyproof.
% \cite{Li17} shows that the notion of OSP
% captures, in some precise sense, those mechanisms which can be
% recognized as strategyproof by cognitively limited agents who do not fully
% understand the game they are playing. % and all of its contingencies.
% In particular, for such agents, restricting attention to games of perfect
% information (as we do) is without loss of generality.
% \cite{Li17} shows that if a social choice function is OSP,
% then in particular it is strategyproof, and its strategyproofness can be
% verified 

\subsubsection{Monotonicity of OSP implementability}
\label{sec:OspMon}

A \emph{subdomain} of an environment $E = (Y, (\T_i)_{i\in\agents})$
is simply another environment $E' = (Y', (\T_i')_{i\in\agents})$,
in which $Y'\subseteq Y$ and $\T_i'\subseteq \T_i$ for each
$i\in\agents=[n]$.
A \emph{restriction} of a social choice function $f$ to $E'$
is the social choice function $f|_{E'}$ such that
$f|_{E'}(t_1,\ldots,t_n) = f(t_1,\ldots,t_n)$
for each $(t_1,\ldots,t_n) \in \T_1'\times\ldots\times\T_n'$,
provided that $f(t_1,\ldots,t_n)\in Y'$ for each
$(t_1,\ldots,t_n)\in\T_1'\times\dots\times\T_n'$.

The following lemma is essentially the ``pruning'' technique
pioneered by \cite{Li17} and used in almost all known works on OSP
mechanisms.

\begin{lemma}
  \label{lem:restrictedOsp}
  If a social choice function $f$ over $E$ is OSP, then
  every restriction $f|_{E'}$ of $f$ to a subdomain $E'$ is also OSP.
\end{lemma}
\begin{proof}
  Consider any OSP implementation $\Game$ of $f$ over $E$.
  Construct a mechanism $\Game'$ over $E' = (Y', (\T_i')_{i\in\agents})$
  with the same set of states,
  the same player function $\mathtt{Pl}$, and the same outcome function
  $g$. Define type sets $\T_i'(h) = \T_i' \cap \T_i(h)$ for each node
  $h$ of $\Game'$.
  % , and prune the resulting game tree as in~\cite{Li17}\footnote{
  %   Briefly, this proceeds as follows:
  %   Note that $\Game'$ may have nodes $h$ where $\T_i'(h) = \emptyset$.
  %   These (and all edges adjacent to them) can be removed 
  %   without impacting $\Game'$. 
  %   After this is done, there may be nodes $h$ with $|\sigma(h)| = 1$.
  %   These can be removed (with the parent edges now
  %   pointing to the unique successor of $h$).
  %   Now, there may be nodes $h$ and $h'\in\sigma(h)$ 
  %   has $\mathtt{Pl}(h')= \mathtt{Pl}(h)$.
  %   These can be ``condensed'' into a single node
  %   whose successors are the collection of nodes $h^*$ for which
  %   there exists a path from $h$ to $h^*$ consisting entirely of nodes $h'$
  %   with $\mathtt{Pl}(h') = \mathtt{Pl}(h)$.
  %   These transformations allow $\Game'$ to satisfy all the simplifying
  %   assumptions of the definition without fundamentally altering the game.
  % }.
  It is clear that $\Game'$ implements $f|_{E'}$.
  Moreover, it follows directly from the definitions that if
  $\Game'$ failed to be OSP, then $\Game$ would fail to be OSP,
  a contradiction.
\end{proof}

% Combining the previous two lemmas, we immediately obtain the following\foo:
% \begin{corollary}
%   
% \end{corollary}

\subsection{Matching mechanisms}

Our primary interest is the \emph{one sided stable matching environment}
with $n$ applicants and $n$ positions.
The outcomes $Y$ of this environment consist of the set of assignments of
all applicants to some (distinct) position.
That is, $Y$ is the set of all bijections $\mu$ between applicants and
positions. We typically denote applicants with letters $a_i$ or 
$\{a,b,c,\ldots\}$, and we denote positions as $p_i$ or with integers
$x\in \{1,2,3,\ldots\}$. We let $\mu(a_i)$ denote the position $a_i$ is assigned
to, and let $\mu(p_i)$ denote the applicant assigned to position $p_i$.

For each applicant $a_i$, the set of types $\T_{a_i}$ consists of all total
orderings over the $n$ positions.
Each applicant is indifferent between matchings which assign the applicant 
to the same position.
Formally, for $t_i \in \T_{a_i}$, we have $\mu \succ_i^{t_i} \mu'$
if and only if $\mu(a_i)\ t_i\ \mu'(a_i)$.
We refer to types of applicants as \emph{preferences} over the positions,
and denote them $\succ_i \in \T_{a_i}$. We typically suppress the notation
$\mu \succ_i^{\succ_i} \mu'$ in favor of simply
$\mu(a_i) \succ_i \mu'(a_i)$.

Consider a set $q = \{ \succ_i \}_{i\in[n]}$ of preferences
over the \emph{applicants}, one for each position.
We refer to these as \emph{priorities},
as they are fixed and not considered strategic.
A matching $\mu$ is \emph{stable} if, for every pair $a, p$
of applicants and positions for which $\mu(a) \ne p$,
we either have $\mu(a) \succ_a p$, or $\mu(p) \succ_p a$.
That is, no pair of one applicant and one position would rather match with
each other than with their assignment under $\mu$.
Note that we restrict attention to complete
preference lists (where every applicant ranks every position and
vice-versa) and environments with the same number of applicants and
positions (and thus, stable matchings must have every applicant and
position matched).

The celebrated results of \cite{GaleS62} are that there always exists a
unique applicant-optimal stable matching (that is, one in which each
applicant is assigned to their favorite position among the set of
positions they are assigned to in any stable matching).
Moreover, this matching can be found in the polynomial time algorithm of
\emph{applicant-proposing deferred acceptance}.
We thus define the social choice function\footnote{
  This is a slight abuse of notion, because the proper deferred acceptance
  algorithm constitutes a mechanism, but we use $\DA^q$ to refer to a social
  choice function.
} $\DA^q$ as follows:
\begin{definition}
For a fixed set of priorities $q = \{ \succ_i \}_i$, we let $\DA^q$
denote the social choice function which, for any set of applicant's 
preferences $\{ \succ_a \}_{a}$, outputs the applicant-optimal
stable matching.
\end{definition}
This social choice function is 
strategyproof~\cite{DubinsMachiavelliGaleShapley81},
% \footnote{
%   As in~\cite{AshlagiG18},
%   the restriction to one sided matching environments is necessary, because 
%   there is no stable social choice function which is strategyproof for both
%   sides of the market~\cite{Roth82}.
%   Moreover, when trying to find OSP stable matching rules,
%   restricting attention to the applicant-optimal stable
%   matching is without loss of generality, because 
%   no other stable matching mechanism is strategyproof for the
%   applicants~\cite{GaleMsMachiavelli85} (and thus other stable matching
%   rules cannot hope to be OSP). 
% },
and the main question this paper addresses is when $\DA^q$ is furthermore
\emph{obviously} strategyproof.

\begin{definition}[OSP priorities]
  We say that a set of priorities $q = \{ \succ_i \}$
  of $n$ positions over $n$ applicants is \emph{OSP implementable} 
  (or simply \emph{OSP} for short)
  if $\DA^q$ is OSP implementable.
\end{definition}

% The central question of this paper is: which priorities $q$ are OSP?

\subsubsection{Monotonicity of OSP priorities}
\label{sec:PriorityMon}

\autoref{sec:OspMon} describes a general monotonicity condition of OSP
mechanisms: restricting attention to a subset of each agent's types
preserves the OSP property.
In the context of matching environments with fixed priorities, we need to
reason about a second kind of monotonicity as well\footnote{
  As demonstrated by~\cite{mandal2020obviously}, this sort of monotonicity
  is somewhat subtle, and does not hold for priorities used in the top
  trading cycles mechanism. However, for deferred acceptance in our
  environment, this type of monotonicity holds.
  % See \autoref{sec:Extensions} for more details on these subtleties.
  % \ctnote{Get the appendix right}
}.

\begin{lemma}
  \label{lem:RestrictedPriorities}
  Consider a set of priorities $q = \{\succ_i\}_i$ of $n$ positions
  over $n$ applicants.
  If $q$ is OSP, then
  for any $m\le n$ and any subset $S$ of $m$ applicants and $T$ of $m$ positions,
  the restriction of $q' = \{ \succ_i|_S \}_{i \in T}$
  is OSP as well.
\end{lemma}
\begin{proof}
  Let $\Game$ be an OSP implementation of $\DA^q$.
  Consider any subdomain $E'$ in which applicants in $S$ rank all positions
  in $T$ before all positions not in $T$, and all applicants not in $S$
  have a fixed preference ranking which puts
  distinct positions not in $S$ as their top preference.
  In a run of deferred acceptance with preferences in $E'$, 
  all applicants outside of $S$ will immediately propose to their top
  choice, and never be rejected.
  Moreover, agents outside of $S$ will have no effect on the run of DA for
  agents in $S$ and preferences in $E'$.
  Thus, consider $\Game'$ restricted to $E'$.
  Restricting $\Game'$ to agents in $S$ give an OSP implementation of
  $\DA^{q'}$.
\end{proof}

\subsubsection{Cyclic priorities}

Cyclic priorities were introduced to the context of studying $\DA^q$
by~\cite{Ergin02}, which categorizes acyclic priorities as precisely those
for which $\DA^q$ always results in a matching which is Pareto optimal for
the applicants.

\begin{definition}[\cite{Ergin02}]
  A priority set $\{ \succ_i \}_i$ over $\agents$ is \emph{cyclic}
  if there exists applicants $a,b,c$ and positions $i,j$
  such that $a \succ_i b \succ_i c$ and $c \succ_j a$.

  Priorities are called \emph{acyclic} otherwise.
\end{definition}

The acyclicity condition is very strong, that is, acyclicity does not allow
for a wide diversity in the priorities. We typically use the
following characterization to describe acyclic priorities:
\begin{proposition}
  \label{prop:AcyclicEquiv}
  Priorities are acyclic if and only if there exists an ordered
  partition $S_1,\ldots,S_k$ of applicants such that
  $1\le|S_i|\le 2$, and
  for all $a \in S_i$ and $b \in S_{i+1}\cup\dots\cup S_{k}$,
  every position $x$ has $a \succ_x b$.
\end{proposition}
\begin{proof}
  If such a partition exists and there is a position $x$ with
  $a \succ_x b \succ_x c$, then $a, c$ cannot be in the same element of the
  partition. Thus, for each position $y$, we have $a \succ_y c$,
  and priorities are acyclic.

  On the other hand, if priorities are acyclic, observe that there can be
  at most two applicants which have the top priority at any
  position. If there is only one such applicant $a$, let $S_1 = \{a\}$.
  If there are two such applicants $a, b$, observe that 
  the top two priorities of each position must be $a$ and $b$
  (if the top two priorities of some position $x$ were $a \succ_x c$,
  then we'd have $a\succ_x c\succ_x b$ and $b \succ_y a$ for some other
  position $y$).
  So let $S_1 = \{a, b\}$.
  Observe that removing $\{a, b\}$ from each priority list results in
  another set of acyclic priorities, so repeating this process until all
  applicants are in some $S_i$ constructs the desired partition.
\end{proof}

\subsubsection{Prior results on OSP implementations of
  \texorpdfstring{$\DA$}{DA}
}
% \subsection{Results of \cite{AshlagiG18}}

\begin{figure}
\tikzset{
  solid node/.style={circle,draw,inner sep=1.5,fill=black},
  decision node/.style={regular polygon, regular polygon sides=6,draw,inner sep=1.5},
  hollow node/.style={circle,draw,inner sep=1.5},
  square node/.style={draw,inner sep=2.5},
  % node gap={xshift=70} 
}
\begin{center}
\begin{tikzpicture}[scale=0.8,font=\footnotesize]
  \tikzstyle{level 1}=[level distance=15mm,sibling distance=19mm]
  \tikzstyle{level 2}=[level distance=15mm,sibling distance=15mm]
  % \tikzstyle{level 3}=[level distance=15mm,sibling distance=7mm]
  \node(a1)[decision node,label=above:{}]{(a)}
    child{node(a11)[square node, label=below:{
        % $\TwoTrader^{2;3}_{b,c}$
      }]{}
      edge from parent node[left,xshift=-3]{$\mathtt{C}(1)$}
    }
    child{node(a12)[square node, label=below:{
        % $\TwoTrader^{1;3}_{b,c}$
      }]{}
      edge from parent node[right,xshift=3]{$\mathtt{C}(2)$}
    };
  \node(b1)[decision node,label=above:{}, xshift=80]{(b)}
    child{node(b11)[hollow node, label=below:{\begin{tabular}{c}
          $\{ (a,3),(b,1),$ \\
          $(c,2) \}$
        \end{tabular} }]{}
      edge from parent node[left,xshift=-3]{$\mathtt{C}(1)$}
    }
    child{node(b13)[square node, label=below:{
        % $\SerialDict^{1,2}_{a,c}$
        \begin{tabular}{c}
          $a$ acts,\\ then $c$.
        \end{tabular}
      }]{}
      edge from parent node[right,xshift=3]{$\mathtt{C}(3)$}
    };
  \node(c1)[decision node,label=above:{$\{1,2\}$ are possible}, xshift=170]{(c)}
    child{node(c11)[hollow node, label=below:{
          \begin{tabular}{c}
          $\{(a,3),(b,2),(c,1)\}$
          \end{tabular} 
        }]{}
      edge from parent node[left,xshift=-3]{$\mathtt{C}(1)$}
    };
  \node(b2)[decision node,label=above:{$\{1,3\}$ are possible}, xshift=270]{(b)}
    child{node(b21)[hollow node, label=below:{
          \begin{tabular}{c}
          $\{(a,3),(b,1),$ \\
          $(c,2)\}$
          \end{tabular} 
        }]{}
      edge from parent node[left,xshift=-3]{$\mathtt{C}(1)$}
    }
    child{node(b23)[square node, label=below:{
        % $\SerialDict^{1,2}_{a,c}$
        \begin{tabular}{c}
          $a$ acts
        \end{tabular}
      }]{}
      edge from parent node[right,xshift=3]{$\mathtt{C}(3)$}
    };

  \draw[->] (a1) -- (b1) 
    node[midway, above]{$3\dots$};
  \draw[->] (b1) -- (c1) 
    node[midway, above]{$2\dots$};
  \draw[->] (c1) -- (b2) 
    node[midway, above]{$2\dots$};
  % \draw[dashed,rounded corners=10]($(1) + (-.2,.25)$)rectangle($(2) +(.2,-.25)$);
  % \node at ($(1)!.5!(2)$) {Bob: $I_B^1$};
\end{tikzpicture}
\end{center}
\caption{
  An OSP implementation of $\DA$ for priorities
  $1 : a\succ b\succ c$, $2: b\succ a\succ c$,
  and $3 : a\succ c\succ b$ as in 
  \autoref{thrm:AshlagiG18}, \autoref{item:AgExceptionalOsp}
  (originally in~\cite[Section 3]{AshlagiG18}).
  Actions labeled with $\mathtt{C}(i)$ indicate an action which ``clinches''
  position $i$ for the applicant acting at that node
  (that is, when an applicant takes action $\mathtt{C}(i)$,
  they well be irrevocably matched to position $i$).
  Square nodes represent subtrees in which the priorities of unmatched
  positions, restricted to unmatched applicants, become acyclic
  (and thus OSP implementable).
  Circular nodes are leaf nodes in which the match is determined.
  Some nodes are labeled above with the set of remaining ``possible''
  positions the applicant may be matched to in all descendant leafs.
  These sets can be calculated using the properties of the classic
  applicant-proposing deferred acceptance algorithm. (For example,
  at the $(c)$ node in the figure,
  we know that $a$'s favorite position is $3$, so $a$
  will propose to $3$ in deferred acceptance. But $a$ has higher priority
  at $3$ than $c$, so $c$ will never be matched to $3$.
  Similar logic implies that at the second $(b)$ node, $b$ cannot get $2$,
  because $c$ has proposed to $2$).
  To verify that this mechanism is OSP, it suffices to check that at the
  second node where any agent acts, they can clinch any position they were
  able to clinch in the first node where they acted.
}
\label{fig:ThreeTraderAshlagiG}
\end{figure}

We concisely sum up the results of \cite{AshlagiG18} as follows:
\begin{theorem}[\cite{AshlagiG18}]
  \label{thrm:AshlagiG18}
  Consider a set of priorities $q$ over applicants.
  \begin{enumerate}
    \item If $q$ is acyclic, then $\DA^q$ is OSP\footnote{
        See \autoref{fig:TwoTrader} for the basic idea of how an OSP
        mechanism is constructed. This is a special case of
        \autoref{thrm:Construction}.
      }.
    \item \label{item:AgExceptionalOsp}
      There exists a $q$ which is cyclic, yet $\DA^q$ is OSP.
      Specifically, one such set of priorities (which is OSP
      implemented by the mechanism described
      in~\autoref{fig:ThreeTraderAshlagiG}) is:
      \begin{align*}
           1 & : \ a \succ b \succ c
        \\ 2 & : \ a \succ c \succ b
        \\ 3 & : \ b \succ a \succ c
      \end{align*}
    \item $\DA^q$ is not OSP in general.
      Specifically, neither of the following preference sets are OSP
      implementable:
      \begin{align*}
           1 & : \ a \succ b \succ c &    1 & : \ a \succ b \succ c
        \\ 2 & : \ b \succ c \succ a &    2 & : \ a \succ b \succ c
        \\ 3 & : \ c \succ a \succ b &    3 & : \ c \succ a \succ b
      \end{align*}
  \end{enumerate}
\end{theorem}
% 
% Two mechanisms from~\cite{AshlagiG18} are very important for this paper.
% We term these mechanisms the ``two-applicant barter'' mechanism
% $\TwoTrader$ and the ``three-applicant lurker'' mechanism $\ThreeTrader$\footnote{
%   For the reader familiar with~\cite{BadeG16}, we now briefly describe why
%   $\ThreeTrader$ is \emph{not} an instance of ``sequential barter with lurkers''
%   from~\cite{BadeG16}. \ctnote{TODO: understand}
%   Nevertheless, we conceptualize $\ThreeTrader$ as a mechanism in which
%   agent $a$ ``lurks'' for position $3$, hoping to achieve it but knowing it
%   can get any other position if it does not succeeded in matching to
%   position $3$.
% }.

% \cite{AshlagiG18} argue (using the logic of
% \autoref{sec:PriorityMon}) that their examples of non-OSP
% implementable priorities indicate that ``most'' real-world priority
% structures will fail to be OSP implementable.
% However, they leave open the question of precisely which priorities are
% OSP.
The main result of this paper is to close the gap between the construction
and impossibility result of~\cite{AshlagiG18}.

\section{Proof that limited cyclic priorities are OSP}
\label{sec:Result}

% While an infinite number cyclic priority sets are OSP, 
% they can only be OSP in very controlled situations.
% 
% \begin{tabular}{c|cccccc}
%   1,2,3 & a  & b  & c  & d  & e \\
%   4 & a  & (c & b) & (e & d) \\
%   5 & (b & a) & (d & c) & e
% \end{tabular}

The cyclic yet OSP example of \autoref{thrm:AshlagiG18},
\autoref{item:AgExceptionalOsp} can be extended in the following way to
include arbitrarily many applicants:

\begin{definition}
  \label{def:limitedCyclic}
  For any $\ell, k \ge 3$, a set of priorities of 
  $ \ell$ positions over applicants $\{a_1, a_2, \ldots, a_k\}$ 
  is \emph{two-adjacent-alternating} if, up to relabeling,
  each priority list is of the form $x = a_1\succ a_2 \succ a_3 \succ a_4 \succ \ldots$
  except for two, one of which is of the form 
  $u$ (where $u$ flips every adjacent pair of applicants starting
  from $a_2, a_3$ onward), and the other is of the form $v$
  (where $v$ flips every adjacent pair of applicants starting from $a_1,
  a_2$ onward). That is, up to relabeling we have
  \begin{align*}
       x &:\ a_1 \succ a_2 \succ a_3 \succ a_4 \succ a_5 \succ a_6 \succ a_7 \ldots
    \\ u &:\ a_1 \succ a_3 \succ a_2 \succ a_5 \succ a_4 \succ a_7 \succ a_6 \ldots
    \\ v &:\ a_2 \succ a_1 \succ a_4 \succ a_3 \succ a_6 \succ a_5 \succ \ldots
  \end{align*}
  where $x$ is repeated $\ell-2$ times, and $u$ and $v$ each appear exactly
  once. % (note that \ctnote{clarify how it ends}).

% \begin{tabular}{c|ccccccccccc}
%   $x$ & $a_1$ &  $a_2$  &  $a_3$  &  $a_4$  & $e$  & $f$   &  $g$  &  $h$  &  $i$  &  \\
%   $u$ & $a$  &  $(c$ &  $b)$ &  $(e$ & $d)$ & $(g$  & $f)$  &  $(i$ &  $h)$  \\
%   $v$ & $b$  &  $a)$ &  $(d$ &  $c)$ & $(f$ & $e)$  & $(h$  &  $g)$ &  $i$
% \end{tabular}
% 
% \begin{tabular}{c|ccccccccccc}
%   $x$ &  $a$  &  $b$  &  $c$  &  $d$  & $e$  & $f$   &  $g$  &  $h$  &  $i$  &  \\
%   $u$ &  $a$  &  $(c$ &  $b)$ &  $(e$ & $d)$ & $(g$  & $f)$  &  $(i$ &  $h)$  \\
%   $v$ &  $b$  &  $a)$ &  $(d$ &  $c)$ & $(f$ & $e)$  & $(h$  &  $g)$ &  $i$
% \end{tabular}

  A set of priorities $\{ \succ_x \}_x$ 
  is \emph{limited cyclic} if there exists an ordered
  partition $S_1, S_2, \ldots, S_L$ of applicants, such that for any
  $a \in S_i$ and $b \in S_{i+1},\ldots,S_{L}$, we have
  $a \succ_x b$ for all positions $x$,
  and moreover, for any $i$ with $|S_i| \ge 3$,
  the restriction of $\{ \succ_x \}_x$ to $S_i$ is
  two-adjacent-alternating.
\end{definition}

For two examples of two-adjacent alternating priorities, consider
\\
\begin{tabular}{c|ccccccccccc}
  1,2,3,4 & $a$ &  $b$  &  $c$  &  $d$  & $e$  & $f$ \\
  5     & $a$ &  $(c$ &  $b)$ &  $(e$ & $d)$ & $f$ \\
  6     & $(b$ &  $a)$ &  $(d$ &  $c)$ & $(f$ & $e)$
\end{tabular}
\qquad
\begin{tabular}{c|ccccccccccccc}
  1,2,3,4,5 & $a$ &  $b$  &  $c$  &  $d$  & $e$  & $f$ & $g$ \\
  6     & $a$ &  $(c$ &  $b)$ &  $(e$ & $d)$ & $(g$ & $f)$  \\
  7     & $(b$ &  $a)$ &  $(d$ &  $c)$ & $(f$ & $e)$ & $g$
\end{tabular}
\\
(where we write each priority list as an ordered list, and add parentheses
to highligh where applicants are flipped from their ``normal ordering'').

Observe that by \autoref{prop:AcyclicEquiv}, 
acyclic priorities are a special case of limited cyclic
priorities, in which each $|S_i| \le 2$.
Moreover, the preference list of \autoref{thrm:AshlagiG18}, 
\autoref{item:AgExceptionalOsp} is limited cyclic.
It turns out that limited cyclic priorities are the \emph{only}
priorities that are cyclic and OSP implementable.
In this section, we prove that limited cyclic priorites are OSP
implementable. In the next section, we prove that \emph{only} limited
cyclic priorities are OSP implementable.

\begin{figure}
\tikzset{
  solid node/.style={circle,draw,inner sep=1.5,fill=black},
  decision node/.style={regular polygon, regular polygon sides=6,draw,inner sep=1.5},
  hollow node/.style={circle,draw,inner sep=1.5},
  square node/.style={draw,inner sep=2.5},
  % node gap={xshift=70} 
}
\hspace{-7mm}
\begin{tikzpicture}[scale=0.8,font=\footnotesize]
  \tikzstyle{level 1}=[level distance=20mm,sibling distance=35mm]
  \tikzstyle{level 2}=[level distance=15mm,sibling distance=15mm]
  % \tikzstyle{level 3}=[level distance=15mm,sibling distance=7mm]
  \node(a1)[decision node,label=above:{}]{(a)}
    child{node(b-sub)[decision node]{(b)}
      child{node(a11)[square node, label=below:{ }]{}
          % \begin{tabular}{c}
          %   continue
          % \end{tabular} 
        edge from parent node[left]{
          \begin{tabular}{l}
            $\mathtt{Clinch}$ any \\ $v \in U\cup V\setminus \{u\}$
          \end{tabular}
        }
      }
      edge from parent node[left]{
        \begin{tabular}{l}
          $\mathtt{Clinch}$ \\ any $u \in U$
        \end{tabular}
      }
    };
  \node(b1)[decision node,label=above:{}, xshift=150]{(b)}
    child{node(b13)[decision node]{(a)}
          % \begin{tabular}{c}
          %   $a$ clinches some $x\in U\cup V\setminus \{v\}$, \\ then continue
          % \end{tabular}
      child{node(a-sub)[square node]{}
        edge from parent node[left]{
          \begin{tabular}{l}
            $\mathtt{Clinch}$ any \\ $u \in U\cup V\setminus \{v\}$
          \end{tabular}
        }
      }
      edge from parent node[left]{\begin{tabular}{l}
          $\mathtt{Clinch}$ \\ any $v \in U \cup V$
        \end{tabular}
      }
    };
  % \node(c1)[decision node,label=above:{1,2 are possible}, xshift=200]{(a)}
  %   child{node(c11)[hollow node, label=below:{\begin{tabular}{c}
  %         $\{(a,3),(b,2),(c,1)\}$
  %       \end{tabular} }]{}
  %     edge from parent node[left,xshift=-3]{$\mathtt{C}(1)$}
  %   };
  % \node(b2)[decision node,label=above:{1,3 are possible}, xshift=300]{(b)}
  %   child{node(b21)[hollow node, label=below:{\begin{tabular}{c}
  %         $\{(a,3),(b,1),$ \\
  %         $(c,2)\}$
  %       \end{tabular} }]{}
  %     edge from parent node[left,xshift=-3]{$\mathtt{C}(1)$}
  %   }
  %   child{node(b23)[square node, label=below:{
  %       $\SerialDict^{1,2}_{a,c}$
  %     }]{}
  %     edge from parent node[right,xshift=3]{$\mathtt{C}(3)$}
  %   };

  \draw[->] (a1) -- (b1) 
    node[midway, above]{$V\ldots$};
  % \draw[->] (b1) -- (c1) 
  %   node[midway, above]{$U\dots$};
  % \draw[->] (c1) -- (b2) 
  %   node[midway, above]{$2\dots$};
  % \draw[dashed,rounded corners=10]($(1) + (-.2,.25)$)rectangle($(2) +(.2,-.25)$);
  % \node at ($(1)!.5!(2)$) {Bob: $I_B^1$};
\end{tikzpicture}

\caption{
  The mechanism $\TwoTrader^{U;V}_{a,b}$, in which all positions
  give top priority to either $a$ or $b$,
  and positions in $U$ rank $a$ first, and positions in $V$ rank $b$ first.
  The name $\TwoTrader$ refers to the idea that $a$ and $b$ can trade their
  priority at $U$ and $V$ if they wish.
  An agent should choose to ``clinch'' a position (that is, become
  perminently matched to that position) when that position is
  their favorite among the set of positions which are still possible for
  them to match to.
  This mechanism was used by~\cite{AshlagiG18} to show that all acyclic
  priorities are OSP implementable.
  % only positions in 
  % $U\cup V$ remain unmatched, and priorities are such that positions
  % in set $U$ rank $a$ first, and positions in set $V$ rank $b$ first
  % (among all currently unmatched agents).
}
\label{fig:TwoTrader}

\vspace{0.3in}

% based on:
% http://www.sfu.ca/~haiyunc/notes/Game_Trees_with_TikZ.pdf
\tikzset{
  solid node/.style={circle,draw,inner sep=1.5,fill=black},
  decision node/.style={regular polygon, regular polygon sides=6,draw,inner sep=1.5},
  hollow node/.style={circle,draw,inner sep=1.5},
  square node/.style={draw,inner sep=2.5},
  % node gap={xshift=70} 
}
% \hspace{-0.45in}
\begin{tikzpicture}[scale=1.3,font=\footnotesize]
  \tikzstyle{level 1}=[level distance=18mm,sibling distance=17mm]
  % \tikzstyle{level 2}=[level distance=15mm,sibling distance=15mm]
  % \tikzstyle{level 3}=[level distance=15mm,sibling distance=7mm]
  \node(a1)[decision node,label=above:{}]{($a_1$)}
    child{node(a11)[square node, label=below:{
          $\ThreeTrader^{X\setminus\{x\},v,u}_{a_2,\ldots,a_k}$
      }]{}
      edge from parent node[left,xshift=5]{
        \begin{tabular}{ccc}
          $\mathtt{Clinch}$ \\
          any \\
          $x\in X$
        \end{tabular}
      }
    }
    child{node(a12)[square node, label=below:{
        acyclic
      }]{}
      edge from parent node[right,yshift=-4,xshift=-4]{
        % $\mathtt{C}(2)$
        \begin{tabular}{ccc}
          $\mathtt{Clinch}$  \\
          $u$
        \end{tabular}
      }
    };

  \node(b1)[decision node,label=above:{}, xshift=110]{($a_2$)}
    child{node(b11)[square node, label=below:{
        \begin{tabular}{c}
          $a_1$ at $v$; \\
          others acyclic
        \end{tabular} 
      }]{}
      edge from parent node[left,xshift=0]{
        \begin{tabular}{ccc}
          $\mathtt{Clinch}$  \\
          any \\
          $x\in X$
        \end{tabular}
      }
    }
    child{node(b13)[square node, label=below:{
        \begin{tabular}{ccc}
          $a_1$ acts; \\
          others acyclic
        \end{tabular}
      }]{(i)}
      edge from parent node[right,xshift=-5,yshift=-2]{
        \begin{tabular}{ccc}
          $\mathtt{Clinch}$  \\
          $v$
        \end{tabular}
      }
    };

  \node(c1)[decision node,label=above:{
        \begin{tabular}{ccc}
          $X\cup\{u\}$  \\
         are possible
        \end{tabular}
      }, xshift=200]{($a_3$)}
    child{node(c11)[square node, label=below:{
        \begin{tabular}{ccc}
          $a_1$ at $v$; \\
          $a_2$ at $u$; \\
          others acyclic
        \end{tabular}
      }]{}
      edge from parent node[left,xshift=3]{
        \begin{tabular}{ccc}
          $\mathtt{Clinch}$  \\
          any \\
          $x\in X$
        \end{tabular}
      }
    };

  \node(b2)[decision node,label=above:{
      \begin{tabular}{cc}
        $X\cup\{v\}$ \\ are possible
      \end{tabular}
    }, xshift=295]{($a_2$)}
    child{node(b21)[square node, label=below:{
        \begin{tabular}{c}
          $a_1$ at $v$; \\
          $a_3$ at $u$; \\
          others acyclic
        \end{tabular} 
      }]{}
      edge from parent node[left,xshift=-3]{
        \begin{tabular}{c}
          $\mathtt{Clinch}$ \\
          any $x\in X$
        \end{tabular} 
      }
    }
    child{node(b23)[square node, label=below:{
        \begin{tabular}{ccc}
          $a_1$ acts; \\
          then $a_3$ acts; \\
          others acyclic
        \end{tabular}
      }]{(ii)}
      edge from parent node[right,xshift=-3]{
        \begin{tabular}{c}
          $\mathtt{Clinch}$ \\
          $v$
        \end{tabular} 
      }
    };

  \draw[->] (a1) -- (b1) 
    node[midway, above]{$v\dots$};
  \draw[->] (b1) -- (c1) 
    node[midway, above]{$u\dots$};
  \draw[->] (c1) -- (b2) 
    node[midway, above]{$u\dots$};
  % \draw[dashed,rounded corners=10]($(1) + (-.2,.25)$)rectangle($(2) +(.2,-.25)$);
  % \node at ($(1)!.5!(2)$) {Bob: $I_B^1$};
\end{tikzpicture}

\caption[Three Trader]{
  The mechanism $\ThreeTrader^{X,u,v}_{a_1,\ldots,a_k}$
  for two-adjacent-alternating priorities over $a_1, a_2, \ldots,a_k$
  (labeled as in \autoref{def:limitedCyclic}, with $X$ denoting the set of
  positions with priority list $x$).
  The name $\ThreeTrader$ refers to the idea that applicant $a_1$ becomes a
  ``lurker'' for position $v$ (following the terminology of \cite{BadeG16}),
  while the mechanism must query $a_2$ and $a_3$ to know if $a_1$ matches
  to $v$.
  If $a_1$ clinches some $x\in X$, then the remaining priorities on
  $a_2,\ldots,a_k$ are also two-adjacent-alternating,
  and we can recursively construct an OSP mechanism.
  In all other branches, one of $u$ or $v$ is matched, and the 
  remaining priorities are acyclic.
  % If any agent is matched to $u$ or $v$, the remaining priorities
  % are acyclic.
  
  All nodes are OSP, by similar logic to that of
  \autoref{fig:ThreeTraderAshlagiG}. Specifically, at the second node where
  any agent acts, they are able to clinch everything they were offered to
  clinch at their first node.
  In particular, when $a_1$ acts in the square node (i),
  they can clinch any position in $X \cup \{u\}$.
  When $a_3$ acts in the square node (ii), two
  things are possible: $a_1$ may have clinched $u$ (in which case $a_3$ can
  clinch any $x\in X$), or not (in which case $a_3$ will match with $u$,
  their favorite attainable position).
  % $1 : a\succ b\succ c$, $2: b\succ a\succ c$,
  % and $3 : a\succ c\succ b$.
  % Actions labeled with $\mathtt{C}(i)$ indicate an action which clinches
  % position $i$ for the applicant acting at that node.
  % Square nodes represent ``sub-mechanisms'' which $\ThreeTrader$
  % recursively runs, and circular nodes represent leaf nodes
  % in which the match is determined.
  % Some nodes are labeled above with the set of remaining ``possible''
  % positions the applicant may be matched to in all descendant leafs.
  % \ctnote{TODO: explain why we know these ``possible sets'' are the way
  % they are.}
  %% This mechanism is adapted from a mechanism constructed
  %% by~\cite{AshlagiG18} to demonstrate that acyclicity is not necesary for
  %% OSP implementability.
}
\label{fig:ThreeTrader}

\end{figure}

\begin{theorem}
  \label{thrm:Construction}
  For any set of priorities $q$ which is limited cyclic,
  $\DA^q$ is OSP implementable.
\end{theorem}
\begin{proof}
  Suppose $q$ is limited cyclic, and let $S_1, S_2, \ldots, S_L$ be the
  partition of applicants as in \autoref{def:limitedCyclic}.
  We use induction on $L$ to construct a mechanism $\Game$ 
  which is OSP and implements $\DA^q$.

  The proof will procede as in~\cite{AshlagiG18} in the case that
  $|S_1|\in\{1,2\}$. Moreover, when $|S_1| = k \ge 3$, the mechanism
  of \autoref{fig:ThreeTraderAshlagiG}
  can be used ``in sequence'' $k-2$ times to assign applicants in $S_1$.
  The key observation is that in \autoref{def:limitedCyclic},
  removing $a_1$ from a two-adjacent-alternating set of priorities results
  in a two-adjacent-alternating set of priorities with one fewer applicant.
  For completeness, we now provide the full details.

  Suppose $|S_1|=1$, say $S_1 = \{a\}$.
  In this case, $\Game$ can simply have $a$ choose to be matched to any
  position, remove $a$ and their match from consideration, and run a
  mechanism OSP implementing $\DA^q$ restricted to $S_2,\ldots,S_L$.
  This is clearly OSP for $a$, and implements $\DA^q$.

  Suppose $|S_1|=2$, say $S_1 = \{a, b\}$.
  Let $U$ be the set of positions which rank $a$ first, and let $V$ be
  those which rank $b$ first.
  Observe that in $\DA^q$, if $a$'s favorite position is any $u \in U$,
  then they will be matched, and if $b$'s favorite position is any $v \in V$,
  then they will be matched. 
  If neither of the above happen, then $a$'s favorite lies in $V$ and $b$'s
  favorite lies in $U$, so both applicants $a$ and $b$ are matched to their
  favorite position.
  Consider the mechanisms $\TwoTrader^{U,V}_{a,b}$ in \autoref{fig:TwoTrader}.
  By the above logic, this mechanism correctly computes the match of $a$
  and $b$. Moreover, this mechanism is OSP for both $a$ (because if $a$
  passes at the root node, then $a$ will either be offered all of $U$ or
  all of $V$) and $b$ (because $b$ acts only once in 
  any execution of the mechanism).
  After running $\TwoTrader^{U,V}_{a,b}$, remove $a, b$ from the matching
  and run an OSP mechanism to match $S_2,\ldots,S_L$ by induction.
  This is OSP and implements $\DA^q$.

  Suppose $|S_1| = k \ge 3$, say $S_1 = \{a_1, a_2, a_3,\ldots,a_k\}$,
  where we label $(a_i)_{i\in[k]}$, $x$, $u$, and $v$
  as in the definition of two-adjacent-alternating priorities
  (\autoref{def:limitedCyclic}).
  Let $X$ denote the set of positions with preference $x$.
  By the same logic of \autoref{fig:ThreeTraderAshlagiG} (which describes a
  mechanism from~\cite{AshlagiG18}),
  we can see that $\ThreeTrader_{a_1,\ldots,a_k}^{X,u,v}$ in
  \autoref{fig:ThreeTrader} correctly computes the match of
  $a_1, a_2,$ and $a_3$.
  Moreover, if the top choice of $a_1$ is some position $x \in X$, then
  $a_1$ can be matched to $x$ and the remaining priorities over 
  $a_2, \ldots, a_k$ are still two-adjacent-alternating. 
  So a mechanism can be recursively constructed in this case.

  The resulting mechanism is OSP for the same reason as
  in~\cite{AshlagiG18}.
  For completeness, we check the OSP constraints here.
  If $a_1$ ``passes'' in the first round (i.e. if they choose not to clinch
  any position), the only way $a_1$ moves again is
  if $a_2$ clinches $v$ at some point. But the next action of $a_1$ always
  allows $a_1$ to clinch $X\cup\{u\}$, which is everything $a_1$ could
  clinch in their first node.
  If $a_2$ ``passes'' in their first node displayed in
  \autoref{fig:ThreeTrader}, then at their second node they can still
  clinch all of $X\cup\{v\}$, which is everything $a_2$ could
  clinch in their first node.
  If $a_3$ passes, then they only move again if $a_2$ clinches $v$ and then
  $a_1$ clinches $u$. But in this case, $a_3$ can clinch any position in
  $X$, which is everything $a_2$ could clinch in their first node.
  Thus, the mechanism is OSP.

  % \autoref{def:limitedCyclic}, which all have the same priority list, and
  % let $u$ and $v$ be 

  After running $\ThreeTrader^{X,u,v}_{a_1,\ldots,a_k}$, remove $S_1$
  from the matching and run an OSP mechanism to match $S_2,\ldots,S_L$ by induction.
  This is OSP and implements $\DA^q$.

  Thus, all limited cyclic priorities are OSP implementable.
\end{proof}
% \subsection{Intuition}
% 
% \paragraph{Mechanism \texorpdfstring{$\ThreeTrader$}{3Lu} cannot be
%   modified to similar priorities}
% 
% \paragraph{Mechanism \texorpdfstring{$\ThreeTrader$}{3Lu} cannot be
%   extended to add more active agents}

\section{Proof of main theorem}
\label{sec:Proof}

% \section{Proof}

The crux of our characterization is
proving that priorities which are not limited cyclic are not OSP.
Formally, this breaks down into two results as follows:
\begin{theorem}
  \label{thrm:figSetsNotOsp}
  None of the priority sets shown in \autoref{fig:AllPrimitiveNonOsp} are
  OSP implementable.
\end{theorem}
\begin{proof}
  For each of Subfigures (a) through (e) of
  \autoref{fig:AllPrimitiveNonOsp} respectively,
  this is shown in \cite[Section 4]{AshlagiG18},
  \autoref{sec:NonOSPTwoWomenSame}, 
  \autoref{sec:NonOspMildExtension}, 
  \autoref{sec:NonOspLastCase}, 
  and \autoref{sec:nIs4CasesProof}.
\end{proof}
\begin{theorem}
  \label{thrm:patternMatchingMainThrm}
  Consider any priority set $q$.
  If no restriction of $q$ is equal to any of the priority sets of
  \autoref{fig:AllPrimitiveNonOsp}, then $q$ is limited cyclic.
\end{theorem}

We prove both these theorems in full in \autoref{sec:FormalProof}.
The proof is organized as outlined in \autoref{sec:intuition}
to naturally reveal how this characterization 
(and in particular, the list of priority sets of
\autoref{fig:AllPrimitiveNonOsp}) was originally attained. 
With these results in hand, the proof of our main characterization is easy:

\begin{theorem*}[Main Theorem]
  For any set of priorities $q$,
  the following are equivalent:
  \begin{enumerate}
    \item Deferred acceptance with priorities $q$ is OSP implementable.
    \item $q$ is limited cyclic.
    \item An OSP mechanism for deferred acceptance with priorities 
      $q$ can be constructed from compositions of
      $\TwoTrader$ and $\ThreeTrader$.
    \item No restriction of $q$ is equal, up to relabeling, to any of the
      priority sets exhibited in \autoref{fig:AllPrimitiveNonOsp}.
  \end{enumerate}
\end{theorem*}
\begin{proof}
  \autoref{thrm:Construction} shows that (2) implies (1) and (3).
  (3) implies (1) by definition.
  By \autoref{thrm:figSetsNotOsp} and \autoref{lem:RestrictedPriorities},
  (1) implies (4).
  By \autoref{thrm:patternMatchingMainThrm}, (4) implies (2).
  This proves the theorem.
\end{proof}

% Items (a)-(d) of \autoref{fig:AllPrimitiveNonOsp} turn
% out to be all cyclic priority sets for three 
% all 
% While the priority lists displayed in \autoref{fig:AllPrimitiveNonOsp} may
% seem to arise from nowhere, 

% \bibliographystyle{ACM-Reference-Format}
\bibliographystyle{alpha}
\bibliography{MechDesign}{}

\clearpage 

\appendix

\section{Key technical lemma}
\label{sec:KeyTechnical}

Papers on obvious strategyproofness typically use \autoref{lem:restrictedOsp}
to find a more tractable subdomain in which to prove a social choice
function is not OSP.
Intuitively, the goal is to find a subdomain in which a mechanism can
\emph{never} ask an agent a nontrivial question about their type 
in an OSP manner\footnote{
  This is equivalent to the ``pruning principle'' in \cite{Li17}.
}.
When we prove a fixed priority list is non-OSP,
we always use the following lemma:
% The main way which \cite{AshlagiG18} (as well as the present paper and
% many subsequent works on OSP implementability)
% uses \autoref{lem:restrictedOsp} this argument is the following:
\begin{lemma}
  \label{lem:keyTechnical}
  Consider a social choice function $f$ over an environment 
  $E = (Y, \T_1,\ldots,\T_n)$.
  Suppose there is a subdomain $E' = (Y, \T_1', \ldots, \T_n')$ of $E$, 
  where $|\T'_i| \in \{1, 2, 3\}$ for each $i$,
  and there is some $i$ for which $|\T_i'|>1$.
  Suppose further that we have:
  % \footnote{
  %   Note that in both this items, we must have $\succ_{-i}, \succ_{-i}'
  %   \in \T_{-i}'$, not mealy $\T_{-i}$, because
  % }
  \begin{enumerate}
    \item \label{item:TwoTypesCase}
      For each $i$ where $|\T'_i|=2$: there exists 
      $\succ_i, \succ_i' \in \T'_i$, 
      where $\succ_i \ne \succ_i'$, and $\succ_{-i}, \succ_{-i}'\in \T'_{-i}$
      such that $f(\succ_i', \succ_{-i}') \succ_i f(\succ_i, \succ_{-i})$. 
    \item \label{item:ThreeTypesCase}
      For each $i$ where $|\T'_i|=3$: for every $\succ_i \in \T'_i$, 
      there exists $\succ_i' \in \T'_i$ with $\succ_i'\ne\succ_i$,
      and $\succ_{-i}, \succ_{-i}'\in \T'_{-i}$
      such that $f(\succ_i', \succ_{-i}') \succ_i f(\succ_i, \succ_{-i})$. 
  \end{enumerate}
  Then $f$ is not OSP.
\end{lemma}
\begin{proof}
  Suppose we have a mechanism $\Game$ implementing $f|_{E'}$
  over the subdomain $E'$.
  Recall that for any node $h$ of $\Game$, we let $\successors(h)$ denote
  the set of immediate successor nodes of $h$.
  For this proof, let $\descendants(h)$ denote the set of
  descendants of $h$ in the game tree.
  Consider the (unique) earliest node $h_0$ in $\Game$ with at least two
  distinct successors $h_1, h_2 \in \successors(h_0)$
  such that $\T_i(h_1), \T_i(h_2)\ne \emptyset$.
  Let $i = \mathtt{Pl}(h_0)$, and observe that for every $j\ne i$, we have
  $\T_j(h_0) = \T_j'$.
  % and recall that every node in $\Game$ as at least two successors,
  % and no type set $\T'_i(h)$ is every empty in $\Game$.
  % Consider the root note $h_0$ of $\Game$,
  % Observe that, because we assume that each node has at least two
  % successors, $|\T_i'(h_0)|>1$.

  Suppose that $|\T'_i|=2$, and let $\T'_i = \{\succ_i, \succ_i'\}$,
  and $\succ_{-i}, \succ_{-i}'\in \T'_{-i}$ be as in \autoref{item:TwoTypesCase}.
  The only possible
  nontrivial partition of $\T'_i$ is $\{ \{ \succ_i\}, \{\succ_i'\} \}$,
  and thus we must have $\successors(h_0) = \{ h_1, h_2 \}$, where 
  $\T'_i(h_1) = \{\succ_i \nobreak \}, \T'_i(h_2) = \{\succ_i'\}$.
  In order for $\Game$ to correctly compute $f$, there must be some
  leaf $\ell \in \descendants(h_1)$ with $g(\ell) = f(\succ_i, \succ_{-i})$,
  and some $\ell' \in \descendants(h_2)$ with 
  $g(\ell') = f(\succ_i', \succ_{-i}')$.
  But then we have $g(\ell') \succ_i g(\ell)$, so $\Game$ is not OSP.

  Suppose that $|\T_i'|=3$. Any nontrivial partition of
  $(\T'_i(h))_{h\in\successors(h_0)}$ must put some type $\succ_i$ in its own set
  of the partition. Thus, $\T'_i(h_1) =\{\succ_i\}$ for some
  $h_1\in\successors(h)$.
  For this $\succ_i$, let $\succ_i'\in \T'_i$ and 
  $\succ_{-i}, \succ_{-i}' \in \T'_{-i}$ be as in
  \autoref{item:ThreeTypesCase}.
  Let $h_2 \in \successors(h_0)$ have $\succ_i' \in \T'_i(h_2)$.
  In order for $\Game$ to correctly compute $f$, there must be some
  leaf $\ell \in \descendants(h_1)$ with $g(\ell) = f(\succ_i, \succ_{-i})$,
  and some $\ell' \in \descendants(h_2)$ with 
  $g(\ell') = f(\succ_i', \succ_{-i}')$.
  But then we have $g(\ell') \succ_i g(\ell)$, so $\Game$ is not OSP.

  Thus, in either case $\Game$ is not OSP\footnote{
    This lemma can readily be extended to $4$ or more types per agent, but
    the resulting case analysis becomes more involved. At worst, one needs
    to look at all possible partitions of $\T'_i$, and find two types 
    $\succ_i, \succ_i'$ in different elements of the partition such that
    $f(\succ_i', \succ_{-i}') \succ_i f(\succ_i, \succ_{-i})$
    for some $\succ_{-i}, \succ_{-i}'$.
    It is interesting to note that all concrete proofs of 
    non-OSP implementability that we are aware of use at most 
    $3$ types per agent.
  }, and no mechanism implementing $f|_{E'}$ can be OSP.
  But by \autoref{lem:restrictedOsp}, this means that no mechanism
  implementing $f$ can be OSP either.
  % But, by \autoref{item:TwoTypesCase}, we have
  % $\Game(\succ_i', \succ_{-i}') \succ_i \Game(\succ_i, \succ_{-i})$.
\end{proof}

\section{Proof of Theorems \ref{thrm:figSetsNotOsp} and
  \ref{thrm:patternMatchingMainThrm} }
\label{sec:FormalProof}

% We have already shown that all limited cyclic priorities are OSP.
In this appendix, we complete the proof of our main theorem by showing that
all priorities which are not limited cyclic are non-OSP.
The proof proceeded through a series of case analysis and 
proving that $\DA^q$ is non-OSP (using \autoref{lem:keyTechnical})
for specific instances of $q$ identified by
the case analysis.

\paragraph{Notation:}
For brevity, we typically write sets of priorities of the positions lists
like \\
\begin{tabular}{c|ccc}
  1 & a& b& c \\
  2 & b& c& a \\
  3 & c& a& b
\end{tabular}
instead of the more cluttered 
\begin{tabular}{cccc}
  $1: a\succ b\succ c$ \\
  $2: b\succ c\succ a$ \\
  $3: c\succ a\succ b$
\end{tabular}. \\
We typically write the preference lists of applicants 
like $p_a^i = 1,2,3$, which indicates that applicant $a$ with
type $p_a^i$ has top preference 
for position $1$, followed by $2$ and then $3$.

We often transcribe the run of applicant-proposing deferred acceptance,
both to see what the outcome is and track what facts about priorities
are important for the argument.
We write this like \\
\begin{tabular}{c|c|c|c|c}
  1 & b c & & a & \\
  2 &   &   &   & b \\
  3 & a & c &
\end{tabular}, \\
where the positions are 1,2,3, their row shows the applicants proposing
to them, and the vertical bars separate proposals that
happen in different ``time steps''.
When we apply \autoref{lem:keyTechnical}, we construct a subdomain
of preferences of applicants, and show that each required case (where some
applicant acts at the root node of the mechanism) is satisfied.
For clarity, we show which cases are the ``truth getting a bad
result'' (which computes $\DA^q(\succ_i,\succ_{-i})$ as in
\autoref{lem:keyTechnical})
and which are a ``lie getting a good result'' (which computes
$\DA^q(\succ_i', \succ_{-i}')$).
% In each such case (or subcase when
% $|\T_i'|=3$) we need 

\subsection{Enumeration for three applicants and positions}
\label{sec:nIs3CasesList}

We start by enumerating all possible priorities of $3$ positions
$\{1,2,3\}$ over $3$ applicants $\{a,b,c\}$, up to relabeling. 
In \autoref{sec:nIs3CasesProof}, we classify all those cases not already
handled by \cite{AshlagiG18} as either OSP or non-OSP.

First, fix the first position's priority as $a \succ b \succ c$ and 
consider the other two positions.
There are 6 cases where the second position also has
priority list $a\succ b\succ c$ (items 1 and 2,
which each contain $3$ distinct cases). 
There are $10 = \binom{5}{2}$ priority sets
where all positions' priority lists are distinct, but it turns
out that there are only four cases up to relabeling (items 3 to 6,
which each contain one case up to relabeling).

\begin{enumerate}
  \item Acyclic cases. These are OSP, as shown in 
    \cite[Section 3]{AshlagiG18} (this is a special case
    of \autoref{thrm:Construction}).
    These are:
    \\
    \begin{tabular}{c|ccc}
      1 & a & b & c \\
      2 & a & b & c \\
      3 & a & b & c
    \end{tabular}
    and
    \begin{tabular}{c|ccc}
      1 & a & b & c \\
      2 & a & b & c \\
      3 & b & a & c
    \end{tabular}
    and
    \begin{tabular}{c|ccc}
      1 & a & b & c \\
      2 & a & b & c \\
      3 & a & c & b
    \end{tabular}.
  \item Cyclic cases where positions 1 and 2 have identical priorities.
    By adapting a proof from \cite[Appendix B]{AshlagiG18},
    we show these are not OSP in \autoref{sec:NonOSPTwoWomenSame}.
    These are: \\
    \begin{tabular}{c|ccc}
      1 & a& b& c \\
      2 & a& b& c \\
      3 & c& a& b
    \end{tabular}
    and
    \begin{tabular}{c|ccc}
      1 & a& b& c \\
      2 & a& b& c \\
      3 & c& b& a
    \end{tabular}
    and
    \begin{tabular}{c|ccc}
      1 & a& b& c \\
      2 & a& b& c \\
      3 & b& c& a
    \end{tabular}.

  \item The ``fully cyclic'' case. 
    This is shown
    to not be OSP in \cite[Section 4]{AshlagiG18}.
    This is: \\
    \begin{tabular}{c|ccc}
      1 & a& b& c \\
      2 & b& c& a \\
      3 & c& a& b
    \end{tabular}.
  \item The cyclic but OSP case.
    These are all equivalent up to relabeling.
    An OSP mechanism for these is constructed in
    \cite[Section 3]{AshlagiG18} (this is a special case
    of \autoref{thrm:Construction}).
    These are: \\
      \begin{tabular}{c|ccc}
        1 & a& b& c \\
        2 & b& a& c \\
        3 & a& c& b
      \end{tabular}
      and
      \begin{tabular}{c|ccc}
        1 & a& b& c \\
        2 & b& a& c \\
        3 & b& c& a
      \end{tabular}
      and
      \begin{tabular}{c|ccc}
        1 & a& b& c \\
        2 & a& c& b \\
        3 & c& a& b
      \end{tabular}.
  \item The case in which two positions have the same
    top priority applicant (but distinct priority lists). 
    These are all equivalent up to relabeling.
    We show these are not OSP in \autoref{sec:NonOspMildExtension}.
    These are: \\
    \begin{tabular}{c|ccc}
      1 & a& b& c \\
      2 & a& c& b \\
      3 & c& b& a
    \end{tabular}
    and
    \begin{tabular}{c|ccc}
      1 & a& b& c \\
      2 & a& c& b \\
      3 & b& c& a
    \end{tabular}
    and
    \begin{tabular}{c|ccc}
      1 & a& b& c \\
      2 & c& b& a \\
      3 & c& a& b
    \end{tabular}.
  \item The case in which all positions have distinct top priority
    applicants (other than the ``fully cyclic'' case).
    These are all equivalent up to relabeling.
    We prove below that this is not OSP
    in \autoref{sec:NonOspLastCase}.
    These are: \\
      \begin{tabular}{c|ccc}
        1 & a& b& c \\
        2 & b& a& c \\
        3 & c& a& b
      \end{tabular}
      and
      \begin{tabular}{c|ccc}
        1 & a& b& c \\
        2 & b& a& c \\
        3 & c& b& a
      \end{tabular}
      and
      \begin{tabular}{c|ccc}
        1 & a& b& c \\
        2 & c& b& a \\
        3 & b& c& a
      \end{tabular}.
\end{enumerate}

% By this enumeration and classification, a set of priorities of $3$ positions
% over $3$ applicants is OSP
% if and only if it is equivalent, up to relabeling, to one of the following:
% \\
% \begin{tabular}{c|ccc}
%   1 & a & b & c \\
%   2 & a & b & c \\
%   3 & a & b & c
% \end{tabular}
% \qquad
% \begin{tabular}{c|ccc}
%   1 & a & b & c \\
%   2 & a & b & c \\
%   3 & b & a & c
% \end{tabular}
% \qquad
% \begin{tabular}{c|ccc}
%   1 & a & b & c \\
%   2 & a & b & c \\
%   3 & a & c & b
% \end{tabular}
% \qquad
% \begin{tabular}{c|ccc}
%   1 & a& b& c \\
%   2 & a& c& b \\
%   3 & b& a& c
% \end{tabular}
% \\
% On the other hand, such a set of priorities is non-OSP
% if and only if it is equivalent, up to relabeling, to one of the following:
% \\
%   \begin{tabular}{c|ccc}
%     1 & a& b& c \\
%     2 & a& b& c \\
%     3 & c& a& b
%   \end{tabular}
%   \qquad
%   \begin{tabular}{c|ccc}
%     1 & a& b& c \\
%     2 & a& b& c \\
%     3 & c& b& a
%   \end{tabular}
%   \qquad
%   \begin{tabular}{c|ccc}
%     1 & a& b& c \\
%     2 & a& b& c \\
%     3 & b& c& a
%   \end{tabular}
%   \\
%   \begin{tabular}{c|ccc}
%     1 & a& b& c \\
%     2 & b& c& a \\
%     3 & c& a& b
%   \end{tabular}
%   \qquad
%   \begin{tabular}{c|ccc}
%     1 & a& b& c \\
%     2 & b& a& c \\
%     3 & a& c& b
%   \end{tabular}
%   \qquad
%   \begin{tabular}{c|ccc}
%     1 & a& b& c \\
%     2 & a& c& b \\
%     3 & c& b& a
%   \end{tabular}
%   \qquad
%   \begin{tabular}{c|ccc}
%     1 & a& b& c \\
%     2 & b& a& c \\
%     3 & c& b& a
%   \end{tabular}

\begin{figure}
  \begin{tabular}{ccc}
  \begin{subfigure}{80mm}
    \centering
    \begin{tabular}{c|ccc}
      1 & a& b& c \\
      2 & a& b& c \\
      3 & a& b& c
    \end{tabular}
    \quad
    \begin{tabular}{c|ccc}
      1 & a& b& c \\
      2 & a& b& c \\
      3 & b& a& c
    \end{tabular}
    \quad
    \begin{tabular}{c|ccc}
      1 & a& b& c \\
      2 & a& b& c \\
      3 & a& c& b
    \end{tabular}
    \caption{OSP by acyclicity}
  \end{subfigure}
  & 
  \begin{subfigure}{35mm}
    \centering
    \begin{tabular}{c|ccc}
      1 & a& b& c \\
      2 & a& c& b \\
      3 & b& a& c
    \end{tabular}
    \caption{OSP but cyclic case}
  \end{subfigure}
  \end{tabular}

  \caption{OSP cases for three positions and three applicants.}

  \vspace{0.5in}

  \begin{tabular}{cccc}

  \begin{subfigure}{25mm}
    \centering
    \begin{tabular}{c|ccc}
      1 & a& b& c \\
      2 & b& c& a \\
      3 & c& a& b
    \end{tabular}
    \caption{Non-OSP by \\ \cite[Section 4]{AshlagiG18}}
  \end{subfigure}
   &
   ~\qquad\ %
   &

  \multicolumn{2}{c}{
  \begin{subfigure}{80mm}
    \centering
    \begin{tabular}{c|ccc}
      1 & a& b& c \\
      2 & a& b& c \\
      3 & c& a& b
    \end{tabular}
    \quad
    \begin{tabular}{c|ccc}
      1 & a& b& c \\
      2 & a& b& c \\
      3 & c& b& a
    \end{tabular}
    \quad
    \begin{tabular}{c|ccc}
      1 & a& b& c \\
      2 & a& b& c \\
      3 & b& c& a
    \end{tabular}
    \caption{Non-OSP by~\cite[Appendix B]{AshlagiG18} and~\autoref{sec:NonOSPTwoWomenSame}}
  \end{subfigure}
  }
  \\
  \\

  \begin{subfigure}{25mm}
    \centering
    \begin{tabular}{c|ccc}
      1 & a& b& c \\
      2 & a& c& b \\
      3 & c& b& a
    \end{tabular}
    \caption{Non-OSP by\\\autoref{sec:NonOspMildExtension}}
  \end{subfigure}
  &
  \quad
  &

  \quad
  \begin{subfigure}{25mm}
    \centering
    \begin{tabular}{c|ccc}
      1 & a& b& c \\
      2 & b& a& c \\
      3 & c& b& a
    \end{tabular}
    \caption{Non-OSP by\\\autoref{sec:NonOspLastCase}}
  \end{subfigure}
  &

  \end{tabular}

  \caption{Non-OSP cases for three positions and three applicants.}
  \label{fig:AllPrimitiveNonOspThreeByThree}
\end{figure}

The following lemmas immediately follow from inspection of the above cases.
In \autoref{sec:nIs4CasesList}, we use the conditions of this proposition
to identify which cases for four applicants and positions
contain no restriction which is non-OSP.

% \begin{proposition}
%   As set of priorities are non-OSP if and only if one of the following
%   hold:
%   \begin{itemize}
%     \item[\texttt{ThreeFav}] hi
%   \end{itemize}
% \end{proposition}
\begin{lemma} \label{lemma:allFavorites}
  If there is a a set of priorities of $3$ positions over $3$ applicant,
  where each applicant has top priority at some position,
  then the priorities are in \autoref{fig:AllPrimitiveNonOspThreeByThree} (and
  thus are not OSP).
\end{lemma}

\begin{lemma} \label{lemma:repeatedPref}
  If there is a set of priorities of $3$ positions over $3$ applicants
  which is cyclic, and in which two priority lists are identical,
  then the priorities are in \autoref{fig:AllPrimitiveNonOspThreeByThree} (and
  thus are not OSP).
\end{lemma}

\begin{lemma} \label{lemma:reversedPref}
  If there is a set of priorities of $3$ positions over $3$ applicants
  for which two positions have exactly reversed priority lists,
  then the priority set is in \autoref{fig:AllPrimitiveNonOspThreeByThree} (and
  thus are not OSP).
\end{lemma}
% We remark that if none of the previous $3$ lemmas hold, 
% then the priorities are indeed OSP (for the case of $3$ applicants and $3$
% positions).
% We need one more lemma as well:

\begin{lemma} \label{lemma:notTransitive}
  If a set of priorities of $3$ positions of $3$ positions over $3$
  applicants not in \autoref{fig:AllPrimitiveNonOspThreeByThree} (i.e. 
  if the priorities are OSP), then there is a pair of applicants $a, c$
  such that $a$ has higher priority than $c$ at every position.
\end{lemma}

\subsection{Cases with three applicants and positions}
\label{sec:nIs3CasesProof}

In this subsection, we prove that Subfigures~(b) through~(d) of
\autoref{fig:AllPrimitiveNonOspThreeByThree} are non-OSP.
This completes the classification of which priorities are OSP
for three applicants and three positions.

\subsubsection{The cases where two priorities are identical}
\label{sec:NonOSPTwoWomenSame}

Appendix B of \cite{AshlagiG18} shows that priorities \\
\begin{tabular}{c|ccc}
  1 & a & b & c \\
  2 & a & b & c \\
  3 & c & a & b
\end{tabular} \\
are non-OSP. We re-prove this, and point out
that the proof as given (with the same applicants' preferences)
does not require anything about positions $3$'s priority
over applicant $b$.
% also shows that a couple other lists of women's
% are not OSP implementable (because agent 3's preference
% for agent b is irrelevant to the argument).

% This covers all non-OSP cases where two priority lists are the same.

\begin{claim}
  \label{claim:NonOspTwoWomenSame}
  Any priorities such that 
  \begin{align*}
    1: &\ a \succ b \succ c  \\
    2: &\ a \succ b \succ c  \\
    3: &\ c \succ a         
  \end{align*}
  are not OSP.
\end{claim}
\begin{proof}

We apply \autoref{lem:keyTechnical} to the
the subdomain in which applicants may have the following
preferences:
\begin{align*}
  % 1: &a \succ b \succ c 
  & p_a^1 = 3,1,2 && p_b^1 = 1,2,3 && p_c^1 = 1,2,3 \\
  % 2: &a \succ b \succ c 
  & p_a^2 = 3,2,1 && p_b^2 = 2,1,3 && p_c^2 = 1,3,2 \\
  % 3: &c \succ a        
  &               &&               && p_c^3 = 2,3,1
\end{align*}

\begin{itemize}
  \item $a$ acts at the root node.

    ``truth getting a bad result'':
    $p_a^2, p_b^1, p_c^2$:
    \begin{tabular}{c|c|c|c}
      1 & b c & \\
      2 &   &   & a \\
      3 & a & c &
    \end{tabular}.
    \\ This works any time $b\succ_1 c$ and $c\succ_3 a$.

    ``lie getting a good result'': 
    $p_a^1, p_b^2, p_c^2$: 
    \begin{tabular}{c|c}
      1 & c \\
      2 & b \\
      3 & a
    \end{tabular}.
    \\ This works regardless of the priorities.

  \item $b$ acts at the root node.

    ``truth getting a bad result'':
    $p_a^1, p_b^1, p_c^2$:
    \begin{tabular}{c|c|c|c|c}
      1 & b c & & a & \\
      2 &   &   &   & b \\
      3 & a & c &
    \end{tabular}.
    \\ This works any time $b\succ_1 c,\ c\succ_3 a,\ a\succ_1 b$.

    ``lie getting a good result'':
    $p_a^2, p_b^2, p_c^3$:
    \begin{tabular}{c|c|c|c|c}
      1 &  &    &   & b\\
      2 & b c & & a \\
      3 & a & c &
    \end{tabular}.
    \\ This works any time $b\succ_2 c,\ c\succ_3 a,\ a\succ_2 b$.

  \item $c$ acts at the root node.
    There are three subcases, depending on which preference
    of $c$ is ``singled out''.

    \begin{itemize}
      \item $p_c^1$ singled out.

        ``truth getting a bad result'':
        $p_a^2, p_b^1, p_c^1$
        \begin{tabular}{c|c|c}
          1 & b c &  \\
          2 &   & c  \\
          3 & a &  
        \end{tabular}.
        \\ This needs $b\succ_1 c$.

        ``lie getting a good result'':
        $p_a^2, p_b^2, p_c^2$:
        \begin{tabular}{c|c}
          1 & c \\
          2 & b \\
          3 & a
        \end{tabular}.
        \\ This works regardless.

      \item $p_c^2$ singled out.
        This is essentially the same as the previous case,
        and again we only need $b\succ_1 c$ for the case to work.

      \item $p_c^3$ singled out.

        ``truth getting a bad result'':
        $p_a^1, p_b^2, p_c^3$
        \begin{tabular}{c|c|c|c}
          1 &     & & a \\
          2 & b c & &  \\
          3 & a & c &
        \end{tabular}.
        \\ This needs $b\succ_2 c$.

        ``lie getting a good result'':
        $p_a^1, p_b^1, p_c^1$:
        \begin{tabular}{c|c|c}
          1 & b c & \\
          2 &     & c \\
          3 & a
        \end{tabular}.
        \\ This needs $b\succ_1 c$.
    \end{itemize}
\end{itemize}

Note that nowhere in the above case analysis did we require anything
about $3$'s preference over $b$ (indeed, in the runs of DA in the argument,
$b$ never proposes to $3$).

\end{proof}
% Thus, the same argument prove that if 
% the women's preferences are
% \\
% \begin{tabular}{c|ccc}
%   1 & a& b& c \\
%   2 & a& b& c \\
%   3 & c& a& b
% \end{tabular}
% or
% \begin{tabular}{c|ccc}
%   1 & a& b& c \\
%   2 & a& b& c \\
%   3 & c& b& a
% \end{tabular}
% or
% \begin{tabular}{c|ccc}
%   1 & a& b& c \\
%   2 & a& b& c \\
%   3 & b& c& a
% \end{tabular}
% \\
% then DA is still OSP implementable.

% This covers all cases where two of the women's preferences are the same.
% Note that in any case where two women's preferences are the same
% and the preferences are cyclic, we have that DA is not OSP implementable.

\subsubsection{The case where two priorities have the same top applicant}
\label{sec:NonOspMildExtension}

Consider the argument in \autoref{sec:NonOSPTwoWomenSame},
which details the facts about the position's priorities are
necessary for the argument to go through.
Note that $2$'s priority over $b$ and $c$ doesn't effect the
argument too much -- only the ``lie'' case of ``$b$ acts at the root''
and the ``truth'' case of ``$p_c^3$ singled out''
required that $b \succ_2 c$.
We'd like to use this observation to additionally prove that
priorities
\\
\begin{tabular}{c|ccc}
  1 & a& b& c \\
  2 & a& c& b \\
  3 & c& b& a
\end{tabular}
\\
are not OSP-implementable. As it turns out,
we will need one additional preference for $b$ (namely $p_b^3$).
% The analysis will also force $b \succ_3 a$ (so we fix 3's preference
% at $c, b, a$), but and all other preferences will stay the same:
\begin{claim}
  Priorities
  \begin{align*}
    1: &\ a \succ b \succ c  \\
    2: &\ a \succ c \succ b  \\
    3: &\ c \succ b \succ a        
  \end{align*}
  are not OSP.
\end{claim}
We apply \autoref{lem:keyTechnical} to the
the subdomain in which applicants may have the following
preferences:
\begin{align*}
  % 1&: a \succ b \succ c & 
  p_a^1 = 3,1,2 && p_b^1 = 1,2,3 && p_c^1 = 1,2,3 \\
  % 2&: a \succ c \succ b &
  p_a^2 = 3,2,1 && p_b^2 = 2,1,3 && p_c^2 = 1,3,2 \\
  % 3&: c \succ b \succ a &
  && p_b^3 = 2,3,1 && p_c^3 = 2,3,1 \\
\end{align*}
Note that all preferences from \autoref{sec:NonOSPTwoWomenSame} are still
in the subdomain -- we merely add $p_b^3$.
We describe how the argument from \autoref{sec:NonOSPTwoWomenSame}
must change.

\begin{itemize}
  \item $a$ acts at the root node. This argument works as before:
    $a$'s preferences are the same, and $2$'s priority list was irrelevant for
    this case. For this to work requires only that $b\succ_1 c$ and $c\succ_3 a$.

  \item $b$ acts at the root node. This argument now has three subcases.
    % Fortunately, they are not terribly hard.

    \begin{itemize}
      \item $p_b^1$ singled out. This is close to the original
        case of $b$ being at the root, but needs modification
        (it actually gets a bit easier to satisfy).

        ``truth getting a bad result'': (as before)
        $p_a^1, p_b^1, p_c^2$:
        \begin{tabular}{c|c|c|c|c}
          1 & b c & & a & \\
          2 &   &   &   & b \\
          3 & a & c &
        \end{tabular}.
        \\ This works any time $b\succ_1 c,\ c\succ_3 a,\ a\succ_1 b$.

        ``lie getting a good result'':
        $p_a^1, p_b^2, p_c^3$:
        \begin{tabular}{c|c|c|c|c}
          1 &     & b \\
          2 & b c & \\
          3 & a & 
        \end{tabular}.
        \\ This requires $c\succ_2 b$ (as is indeed now the case).

      \item $p_b^2$ singled out.

        ``truth getting a bad result'':
        $p_a^1, p_b^2, p_c^3$:
        \begin{tabular}{c|c|c|c|c}
          1 &     & b \\
          2 & b c & \\
          3 & a &
        \end{tabular}.
        \\ This requires $c\succ_2 b$.

        ``lie getting a good result'':
        $p_a^1, p_b^3, p_c^2$:
        \begin{tabular}{c|c}
          1 & c  \\
          2 & b \\
          3 & a
        \end{tabular}.
        % \\ This requires $b\succ_1 c, c\succ_3 a, a\succ_1 b $.
        \\ This works regardless of priorities.

      \item $p_b^3$ singled out.

        ``truth getting a bad result'':
        $p_a^1, p_b^3, p_c^3$:
        \begin{tabular}{c|c|c|c}
          1 &     & & a \\
          2 & b c & \\
          3 & a & b &
        \end{tabular}.
        \\ This requires that $c\succ_2 b$.

        ``lie getting a good result'':
        (as the previous subcase):
        $p_a^1, p_b^2, p_c^2$:
        \begin{tabular}{c|c}
          1 & c  \\
          2 & b \\
          3 & a
        \end{tabular}.
        % \\ This requires $b\succ_1 c, c\succ_3 a, a\succ_1 b $.
        \\ This works regardless of priorities.

        % ``lie getting a good result'':
        % $p_a^1, p_b^1, p_c^1$, with same requirements.
    \end{itemize}

  \item $c$ acts at the root node. One of the subcases changes.

    \begin{itemize}
      \item $p_c^1$ singled out.
        This case still holds as before. 
        This requires $b\succ_1 c$.

      \item $p_c^2$ singled out.
        This case still holds as before. 
        This requires $b\succ_1 c$.

      \item $p_c^3$ singled out. This subcase needs different
        logic because the original argument depended
        on $b \succ_2 c$ (indeed, fixing this subcase was the reason we need
        to add the preference $p_b^3$ to the analysis).
        % \footnote{
        %   Here we need b to go to 3 after 2, but earlier
        %   in the ``b at root, lying'' case we need b to go to 1 after 2.
        % }

        ``truth getting a bad result'':
        $p_a^2, p_b^3, p_c^3$
        \begin{tabular}{c|c|c|c|c|c}
          1 &     & &   &   & b \\
          2 & b c & & a &   &   \\
          3 & a & b &   & c &
        \end{tabular}.
        \\ This needs $c\succ_2 b, b\succ_3 a, a\succ_2 c$.

        ``lie getting a good result'': (as before)
        $p_a^1, p_b^1, p_c^1$:
        \begin{tabular}{c|c|c}
          1 & b c & \\
          2 &     & c \\
          3 & a
        \end{tabular}.
        \\ This needs $b\succ_1 c$.
    \end{itemize}
\end{itemize}

% As an aside, we can observe that for the argument to work, we need
% $a>_1 b >_1 c$ and $a >_2 c >_2 b$, but for 3's preference, we just need
% $c>_3 a$ and $b>_3 a$. As this is all we need, we have also
% proven that the preference set
% \begin{tabular}{c|ccc}
%   1 & a& b& c \\
%   2 & a& c& b \\
%   3 & b& c& a
% \end{tabular}
% (which flips $3$'s pref for $b$ and $c$) is also not OSP implementable.
% Though we actually could have seen this anyway because this preference
% set is equivalent to 
% \begin{tabular}{c|ccc}
%   1 & a& b& c \\
%   2 & a& c& b \\
%   3 & c& b& a
% \end{tabular}
% up to relabeling.

\subsubsection{The case where all priorities have distinct top applicants}
\label{sec:NonOspLastCase}

We now consider the final case with three positions and three applicants,
namely
\\
\begin{tabular}{c|ccc}
  1 & a& b& c \\
  2 & b& a& c \\
  3 & c& a& b
\end{tabular}.
\\
This has a bit different flavor because every applicant has first priority
at some position.
% I kind of had a feeling that this would mean we would have to consider
% three preferences for each number, but it turns out you can get away with
% only two preferences for agent b.
\begin{claim}
  Priorities
  \begin{align*}
    1: &\ a \succ b \succ c & \\
    2: &\ b \succ a \succ c & \\
    3: &\ c \succ a \succ b &
  \end{align*}
  are not OSP.
\end{claim}
\begin{proof}

We apply \autoref{lem:keyTechnical} to the
the subdomain in which applicants may have the following
preferences:
\begin{align*}
   p_a^1 = 2,1,3 && p_b^1 = 3,2,1 && p_c^1 = 2,3,1 \\
   p_a^2 = 3,1,2 && p_b^2 = 1,3,2 && p_c^2 = 1,2,3 \\
   p_a^3 = 3,2,1 &&               && p_c^3 = 1,3,2 \\
\end{align*}

\begin{itemize}

  \item $a$ acts at the root node.
  \begin{itemize}
    \item $p_a^1$ singled out.

      ``truth getting a bad result'': $p_a^1, p_b^1, p_c^1$:
          \begin{tabular}{c|c|c|c|c}
            1 &     &  &   & a \\
            2 & a c &  & b &    \\
            3 & b   & c &
          \end{tabular}.

      ``lie getting a good result'': $p_a^3, p_b^2, p_c^3$:
          \begin{tabular}{c|c|c|c}
            1 & b c &   & \\
            2 &     &   & a \\
            3 & a   & c &
          \end{tabular}.

    \item $p_a^2$ singled out.

      ``truth getting a bad result'': $p_a^2, p_b^2, p_c^3$:
          \begin{tabular}{c|c|c|c|c|c}
            1 & b c &  & a && \\
            2 &     &  &   &   & b \\
            3 & a   & c&   & b
          \end{tabular}.

      ``lie getting a good result'': $p_a^3, p_b^2, p_c^1$:
          \begin{tabular}{c|c}
            1 & b \\
            2 & c \\
            3 & a
          \end{tabular}.

    \item $p_a^3$ singled out.

      This is essentially the same as the previous case.
      In the ``truth getting a bad result'' $a$ gets kicked out of 3,
      but in the ``lie getting a good result'' b and c leave a alone.
  \end{itemize}

  \item $b$ acts at the root node.

    ``truth getting a bad result'': $p_a^2, p_b^2, p_c^3$:
    \begin{tabular}{c|c|c|c|c|c}
      1 & b c &   & a &   & \\
      2 &     &   &   &   & b\\
      3 & a   & c &   & b &
    \end{tabular}.

    ``lie getting a good result'':  $p_a^1, p_b^1, p_c^3$:
    \begin{tabular}{c|c}
      1 & c  \\
      2 & a \\
      3 & b
    \end{tabular}.

    % Remark: this is the first case which has used a
    % proposing side agent going to their third choice
    % and deviating to get their second choice
    % (as opposed to second and first choice resp.).

  \item $c$ at the root.
  \begin{itemize}
    \item $p_c^1$ singled out.

      ``truth getting a bad result'': $p_a^1, p_b^2, p_c^1$:
          \begin{tabular}{c|c|c}
            1 &   b &   \\
            2 & c a &   \\
            3 &     & c
          \end{tabular}.

      ``lie getting a good result'': $p_a^2, p_b^2, p_c^2$:
          \begin{tabular}{c|c|c}
            1 & b c &    \\
            2 &     & c  \\
            3 & a   &
          \end{tabular}.

    \item $p_c^2$ singled out.

      ``truth getting a bad result'': $p_a^2, p_b^2, p_c^2$:
          \begin{tabular}{c|c|c}
            1 & b c &     \\
            2 &     & c  \\
            3 & a   &
          \end{tabular}.

      ``lie getting a good result'': $p_a^1, p_b^1, p_c^3$:
          \begin{tabular}{c|c}
            1 & c \\
            2 & a \\
            3 & b
          \end{tabular}.

    \item $p_c^3$ singled out.

      This is essentially the same as the previous case.
      In the ``truth getting a bad result'' $c$ gets kicked out of 1,
      but in the ``lie getting a good result'' $b$ and $c$ leave $a$ alone.
  \end{itemize}

\end{itemize}

\end{proof}

\subsection{Enumeration for four applicants and positions}
\label{sec:nIs4CasesList}

We now describe and enumerate all cyclic priority sets for four applicants
and four positions for which no restriction to three applicants and three
positions is non-OSP.
% which do not contain a non-OSP set of prior
Due to the fact that (up to relabeling) 
exactly one set of cyclic priorities over three applicants is OSP, 
the collection of priority sets we need to consider is fairly small.
Our enumeration proceeds through a series of lemmas.

\begin{lemma}
  \label{lem:3by4}
  Consider a set of priorities of $n \ge 3$ positions over $3$ applicants.
  Suppose the priorities are cyclic, yet no set of $3$ positions have
  priorities over the three applicants which are equal (up to relabeling)
  to any of the priority sets in
  \autoref{fig:AllPrimitiveNonOspThreeByThree}
  (i.e. no restriction to $3$ positions is non-OSP).
  The, up to relabeling, every position's priority list is of the form $x$
  except for two, which are of the form $u$ and $v$, where we have
  \\
  \begin{tabular}{c|cccc}
    $x$ & $a$ & $b$ & $c$ \\
    $u$ & $a$ & $c$ & $b$ \\
    $v$ & $b$ & $a$ & $c$
  \end{tabular}
  \\
  That is, priorities $q$ are two-adjacent-alternating.
  % The only preference pattern with four women and three men $a, b, c$
  % which is cyclic, yet does not contain one of the ``forbidden patterns''
  % (which are non-OSP with 3 men and 3 women) is,
  % up to relabeling,
  % \\
  % \begin{tabular}{c|cccc}
  %   1 & a & b & c \\
  %   2 & a & b & c \\
  %   3 & a & c & b \\
  %   4 & b & a & c
  % \end{tabular}
  % Moreover, for $n>4$, the only such patterns repeat preferences 1 and 2 as
  % many times as you like.
\end{lemma}
\begin{proof}
  Such a priority set must, in particular, have an OSP and cyclic
  set of priorities of $3$ positions contained in it.
  By \autoref{sec:nIs3CasesList}, up to relabeling, 
  the only $3$ applicant and $3$ position priority
  set which is OSP is
  \\
  \begin{tabular}{c|cccc}
    $x$ & $a$ & $b$ & $c$ \\
    $u$ & $a$ & $c$ & $b$ \\
    $v$ & $b$ & $a$ & $c$
  \end{tabular}.
  \\
  Thus, by relabeling, we can assume that three position's priority 
  lists are $x, u, v$.
  We now consider what the priority lists of the other positions could
  possibly be.

  By \autoref{lemma:allFavorites}, other positions cannot have $c$ 
  as their top priority.
  Thus, the only possible preference other
  than those of $x, u, v$ is $w : b\succ c\succ a$.
  But then positions $u, v, w$ together form a pattern equivalent
  %1, 3, and 4 
  to that of \autoref{sec:NonOspMildExtension}, and are not OSP.

  Thus, all other positions must have priority lists which are equal to one
  of $x, u, $ or $v$.
  % So woman 1's preference must be a repeat of one of women 2, 3, or 4.
  If the priority list $u$ or $v$ is repeated, say at position $y$, then
  \autoref{lemma:repeatedPref} implies that $u, v$, and $y$
  together form a preference list in
  \autoref{fig:AllPrimitiveNonOspThreeByThree}.
  Thus, the only possibility is repeating the priority list of $x$.

\end{proof}

\begin{lemma}
  \label{lem:4by4Cases}
  Suppose $q$ is a set of priorities of $4$ positions over $4$ applicants.
  Suppose $q$ is cyclic, yet no restriction of $q$ is equal (up to
  relabeling) to any priority set of
  \autoref{fig:AllPrimitiveNonOspThreeByThree}.
  Then, up to relabeling, $q$ must equal one of the following priority sets:
  \\
  \begin{tabular}{ccc}
    \begin{tabular}{c|cccccccc}
      1 & {\bf d} & a & b & c \\
      2 & {\bf d} & a & b & c \\
      3 & {\bf d} & a & c & b \\
      4 & {\bf d} & b & a & c
    \end{tabular}
    &
    \qquad
    \begin{tabular}{c|cccccc}
      1 & a & {\bf d} & b & c \\
      2 & {\bf d} & a & b & c \\
      3 & {\bf d} & a & c & b \\
      4 & {\bf d} & b & a & c
    \end{tabular}
    &
    \qquad
    \begin{tabular}{c|cccccc}
      1 & {\bf d} & a & b & c \\
      2 & {\bf d} & a & b & c \\
      3 & a & {\bf d} & c & b \\
      4 & {\bf d} & b & a & c
    \end{tabular}
    \\
    \\
    \begin{tabular}{c|cccccc}
      1 & a & b & c & {\bf d} \\
      2 & a & b & c & {\bf d} \\
      3 & a & c & b & {\bf d} \\
      4 & b & a & {\bf d} & c
    \end{tabular}
    &
    \qquad
    \begin{tabular}{c|cccccc}
      1 & a & b & c & {\bf d} \\
      2 & a & b & {\bf d} & c \\
      3 & a & c & b & {\bf d} \\
      4 & b & a & c & {\bf d}
    \end{tabular}
    &
    \qquad
    \begin{tabular}{c|cccccc}
      1 & a & b & c & {\bf d} \\
      2 & a & b & c & {\bf d} \\
      3 & a & c & b & {\bf d} \\
      4 & b & a & c & {\bf d}
    \end{tabular}
  \end{tabular}
\end{lemma}

\begin{proof}
Suppose the priorities are cyclic over applicants $a,b,c$.
By \autoref{lem:3by4}, we can assume without loss of generality that the
priorities, restricted to applicants $a, b, c$, are
\\
\begin{tabular}{c|cccccc}
  1 & a & b & c \\
  2 & a & b & c \\
  3 & a & c & b \\
  4 & b & a & c
\end{tabular}
\\
Thus, we need on only consider where applicant $d$ is inserted on the
priority list of each position.
We divide into cases based on the position of applicant $d$ on position
$4$'s list.
Throughout the proof, we make heavy use of Lemmas~\ref{lemma:allFavorites}
through~\ref{lemma:notTransitive}.
% , applied using the results of
% \autoref{lem:RestrictedPriorities}.
To make the proof more readable, we mark this assumption on 4's list by
(d), and when we consider a slot in a priority list where d might be
located, we write ``d?''.

\begin{itemize}
  \item $4\ | \  (d) \  b \  a \  c$. \\
    By \autoref{lemma:reversedPref}, $d$ cannot be after $b$
    in any priority list of $i\in\{1,2,3\}$ (as then we would have
    $d \succ_4 b \succ_4 a$ and
    $a \succ_i b \succ_i d$ for some $i\in\{1,2,3\}$).
    % This leaves $2\cdot 2\cdot 3$ settings of $d$ in each list.
    Thus $d$ can be in the following places:
    \\
    \begin{tabular}{c|cccccc}
      1 & d? & a & d? & b & c \\
      2 & d? & a & d? & b & c \\
      3 & d? & a & d? & c & d? & b \\
      4 & (d) & b & a & c
    \end{tabular}
    \\
    By \autoref{lemma:repeatedPref}, at most one position 
    $i\in\{1,2,3\}$ can have $a \succ_i d \succ_i b$ 
    (as we have $b \succ_4 a$).
    Moreover, if $i=1$ and $2$ both have $d \succ_i a \succ_i c$, then
    we cannot have $c \succ_3 d$ (by \autoref{lemma:repeatedPref}
    applied to $d, a, c$).
    Thus, up to relabeling 1 and 2, the cases are
    \\
    \begin{tabular}{c|cccccc}
      1 & {\bf d} & a & b & c \\
      2 & {\bf d} & a & b & c \\
      3 & {\bf d} & a & c & b \\
      4 & (d) & b & a & c
    \end{tabular}
    \qquad
    \begin{tabular}{c|cccccc}
      1 & a & {\bf d} & b & c \\
      2 & {\bf d} & a & b & c \\
      3 & {\bf d} & a & c & b \\
      4 & (d) & b & a & c
    \end{tabular}
    \qquad
    \begin{tabular}{c|cccccc}
      1 & {\bf d} & a & b & c \\
      2 & {\bf d} & a & b & c \\
      3 & a & {\bf d} & c & b \\
      4 & (d) & b & a & c
    \end{tabular}
    % \qquad
    % \begin{tabular}{c|cccccc}
    %   1 & d & a & b & c \\
    %   2 & d & a & b & c \\
    %   3 & a & c & d & b \\
    %   4 & d & b & a & c
    % \end{tabular}

    \vspace{3mm}

  \item $4\ | \ b \ (d) \  a \  c$.  \\
    By \autoref{lemma:allFavorites}, either all $i\in\{1,2,3\}$
    have $d \succ_i a$, or none of them do.
    Moreover, by \autoref{lemma:reversedPref}, we never
    have $a \succ_i d \succ_i b$
    (as $b \succ_4 d \succ_4 a$).
    Thus, the places d could be are either
    % up to relabeling 1 and 2, the cases are
    \\
    \begin{tabular}{c|cccccc}
      1 & d? & a & b & c \\
      2 & d? & a & b & c \\
      3 & d? & a & c & b \\
      4 & b & (d) & a & c
    \end{tabular}
    \qquad
    or
    \qquad
    \begin{tabular}{c|cccccc}
      1 & a & b & d? & c & d? \\
      2 & a & b & d? & c & d? \\
      3 & a & c & b & d \\
      4 & b & (d) & a & c
    \end{tabular}
    \\
    % We can use \autoref{lemma:repeatedPref} to see that none of these
    % can be OSP.
    For both cases, we can use \autoref{lemma:repeatedPref} (applied
    to positions $1, 2, 4$ and applicants $d, a, b$) to see that all of
    these cases contain a priority set of
    \autoref{fig:AllPrimitiveNonOspThreeByThree}.
    % For the second case above, lemma~\ref{lemma:repeatedPref} applies
    % to women 1, 2, 4 and men $a, b, d$.
    % Up to relabeling 1 and 2, the second case can become
    % \\
    % \begin{tabular}{c|cccccc}
    %   1 & a & b & d & c \\
    %   2 & a & b & d & c \\
    %   3 & a & c & b & d \\
    %   4 & b & (d) & a & c
    % \end{tabular}
    % \qquad
    % \begin{tabular}{c|cccccc}
    %   1 & a & b & d & c \\
    %   2 & a & b & c & d \\
    %   3 & a & c & b & d \\
    %   4 & b & (d) & a & c
    % \end{tabular}
    % \qquad
    % \begin{tabular}{c|cccccc}
    %   1 & a & b & c & d \\
    %   2 & a & b & c & d \\
    %   3 & a & c & b & d \\
    %   4 & b & (d) & a & c
    % \end{tabular}
    % \\
    % All of these fail to be OSP by lemma~\ref{lemma:repeatedPref}.
    % (for the 1st, use $b,d,c$ and
    % for the 2nd and 3rd, use $a, b, d$).
    Thus, there are \textbf{no} possible cases with
    $4\ | \ b \ (d) \  a \  c$.

    \vspace{3mm}

  \item $4\ | \ b \ a \ (d) \ c$.  \\
    By \autoref{lemma:reversedPref}, we never have
    $d \succ_i a \succ_i b$ for $i\in\{1,2,3\}$
    (as we have $b \succ_4 a \succ_4 d$).
    Moreover, we can't have $c \succ_3 d \succ_3 b$
    (as we have $b \succ_4 d \succ_4 c$).
    If we have $d \succ_3 c \succ_3 b$, then by \autoref{lemma:reversedPref}
    we cannot have $b \succ_i c \succ_i d$ for $i\in\{1,2\}$.
    Thus the possible cases are:
    \\
    \begin{tabular}{c|cccccc}
      1 & a & d? & b & d? & c \\
      2 & a & d? & b & d? & c \\
      3 & a & d? & c & b \\
      4 & b & a & (d) & c
    \end{tabular}
    \qquad
    \begin{tabular}{c|cccccc}
      1 & a & d? & b & d? & c & d?\\
      2 & a & d? & b & d? & c & d?\\
      3 & a & c & b & d? \\
      4 & b & a & (d) & c
    \end{tabular}
    \\
    Let's examine the first subcase.
    If either $i\in\{1,2\}$ has $a\succ_i d \succ_i b$,
    then the priorities are not OSP (by \autoref{lemma:repeatedPref}
    applied to applicants $a, d, b$ and positions $i$, 3 and 4).
    But if both $i\in\{1,2\}$ have $b \succ_i d \succ_i c$,
    then the preference is also not OSP
    by \autoref{lemma:repeatedPref}
    (applied to positions $1, 2, 3$ and applicants $b,d,c$).
    So \textbf{none} of the preferences in this subcase can be OSP.

    Now let's consider the second subcase.
    By \autoref{lemma:reversedPref}, we cannot have
    $d \succ_i b \succ_i c$  for either of $i\in\{1,2\}$ 
    (as we have $c \succ_3 b \succ_3 d$).
    Moreover, we also cannot have $b \succ_i d \succ_i c$ 
    for either $i\in\{1,2\}$
    by \autoref{lemma:repeatedPref},
    (applied to position $i, 3, 4$ and applicants $b,d,c$).
    % because $b \succ_4 d \succ_4 c$ but $c \succ_3 b$).
    Thus, the only possible case remaining is
    \\
    \begin{tabular}{c|cccccc}
      1 & a & b & c & {\bf d} \\
      2 & a & b & c & {\bf d} \\
      3 & a & c & b & {\bf d} \\
      4 & b & a & (d) & c
    \end{tabular}

    % ["abcd","abcd","acbd","badc"]
    % ^{ THIS should be the only one}<++>
    %
    % ["adbc","abcd","acbd","badc"]
    % -- ^ symmetric ["abcd","adbc","acbd","badc"]
    % ["abcd","abcd","acbd","badc"]
    %
    % By lemma~\ref{lemma:repeatedPref},
    % $i=1$ and $2$ cannot both have $d >_i b >_i c$
    % (as $c >_3 d$) and they cannot both have
    % $b >_i d >_i c$ (as $c >_3 b$).

    % the 1st is not OSP (applied to $a, d, b$), 
    % and neither is the 2nd (applied to $b, d, c$).
    % So the subcases here are
    % \begin{tabular}{c|cccccc}
    %   1 & a & d & b & c \\
    %   2 & a & d & b & c \\
    %   3 & a & d & c & b \\
    %   4 & b & a & d & c
    % \end{tabular}
    % \qquad
    % \begin{tabular}{c|cccccc}
    %   1 & a & b & d & c \\
    %   2 & a & b & d & c \\
    %   3 & a & c & b & d \\
    %   4 & b & a & d & c
    % \end{tabular}
    % \qquad
    % \begin{tabular}{c|cccccc}
    %   1 & a & b & d & c \\
    %   2 & a & b & c & d \\
    %   3 & a & c & b & d \\
    %   4 & b & a & d & c
    % \end{tabular}
    % \qquad
    % \begin{tabular}{c|cccccc}
    %   1 & a & b & c & d \\
    %   2 & a & b & c & d \\
    %   3 & a & c & b & d \\
    %   4 & b & a & d & c
    % \end{tabular}

    \vspace{3mm}

  \item $4\ | \ b \ a \ c \ (d)$.  \\
    By \autoref{lemma:reversedPref},
    we never have $d \succ_i a \succ_i b$ for any $i\in\{1,2,3\}$
    (as we have $b \succ_4 a \succ_4 d$),
    and moreover we can't have $d \succ_3 c \succ_3 b$
    (as we have $b \succ_4 c \succ_4 d$).
    Thus, the possible locations for $d$ are 
    \\
    \begin{tabular}{c|cccccc}
      1 & a & d? & b & d? & c & d? \\
      2 & a & d? & b & d? & c & d? \\
      3 & a & c & d? & b & d? \\
      4 & b & a & c & (d)
    \end{tabular}
    \\
    By \autoref{lemma:reversedPref},
    if we have $c \succ_3 d \succ_3 b$, then we cannot have 
    $b \succ_i d \succ_i c$ for either of $i\in\{1,2\}$,
    and if $c \succ_3 b \succ_3 d$, then we cannot 
    have $d \succ_i b \succ_i c$ for either of $i\in\{1,2\}$.
    Thus, the cases which can be OSP are:
    \\
    \begin{tabular}{c|cccccc}
      1 & a & d? & b & c & d? \\
      2 & a & d? & b & c & d? \\
      3 & a & c & d & b \\
      4 & b & a & c & (d)
    \end{tabular}
    \qquad
    \begin{tabular}{c|cccccc}
      1 & a & b & d? & c & d? \\
      2 & a & b & d? & c & d? \\
      3 & a & c & b & d \\
      4 & b & a & c & (d)
    \end{tabular}
    
    Consider the first subcase above.
    By \autoref{lemma:repeatedPref} (applied to positions $1,2,3$
    and applicants $b,d,c$), in the first subcase above 
    we cannot have $d \succ_i b \succ_i c$ for \emph{both} $i\in\{1,2\}$.
    But if $b \succ_i c \succ_i d$ for some $i\in\{1,2\}$,
    then the preference is not OSP, again by \autoref{lemma:repeatedPref}
    (applied to positions $i, 3, 4$ and applicants $b,c,d$).
    So \textbf{none} of the priorities in this subcase are OSP.

    Consider now the second subcase above.
    Again by \autoref{lemma:repeatedPref} (applied to positions $1,2,3$
    and applicants $b,d,c$), in the second subcase
    we cannot have $b \succ_i d \succ_i c$ for both $i\in\{1,2\}$.
    Thus, up to relabeling positions $1,2$, the possible OSP cases are
    \\
    % \begin{tabular}{c|cccccc}
    %   1 & a & d & b & c \\
    %   2 & a & b & c & d \\
    %   3 & a & c & d & b \\
    %   4 & b & a & c & d
    % \end{tabular}
    % \qquad
    % \begin{tabular}{c|cccccc}
    %   1 & a & b & c & d \\
    %   2 & a & b & c & d \\
    %   3 & a & c & d & b \\
    %   4 & b & a & c & d
    % \end{tabular}
    % \qquad
    \begin{tabular}{c|cccccc}
      1 & a & b & {\bf d} & c \\
      2 & a & b & c & {\bf d} \\
      3 & a & c & b & {\bf d} \\
      4 & b & a & c & (d)
    \end{tabular}
    \qquad
    \begin{tabular}{c|cccccc}
      1 & a & b & c & {\bf d} \\
      2 & a & b & c & {\bf d} \\
      3 & a & c & b & {\bf d} \\
      4 & b & a & c & (d)
    \end{tabular}
    % The 1st is not OSP, by lemma~\ref{lemma:repeatedPref}
    % ($b >_i c >_i d$ for $i=2,4$ but $d >_3 b$)
    % and neither is the 2nd
    % ($b >_i c >_i d$ for $i=1,2$ but $d >_3 b$).
    % So the remaining cases are (up to relabeling 1 and 2)
    % \\
    % \begin{tabular}{c|cccccc}
    %   1 & a & b & c & d \\
    %   2 & a & b & d & c \\
    %   3 & a & c & b & d \\
    %   4 & b & a & c & d
    % \end{tabular}
    % \qquad
    % \begin{tabular}{c|cccccc}
    %   1 & a & b & c & d \\
    %   2 & a & b & c & d \\
    %   3 & a & c & b & d \\
    %   4 & b & a & c & d
    % \end{tabular}
    % ["abdc","abcd","acbd","bacd"]
    % -- ^^ symmetric ["abcd","abdc","acbd","bacd"]
    % ["abcd","abcd","acbd","bacd"]

\end{itemize}

% One can check that each of the six cases identified indeed contains no
% forbidden 3 by 3 pattern (the author did so with a short computer program).
% However, this is not technically necessary for our proof, so we skip
% verifying it formally.

\end{proof}

% \begin{corollary}
%   Up to relabeling, the only case which contains a cycle,
%   with four men and $\ge 4$ women,
%   such that no man is always below or above every other man,
%   and does not contain a forbidden pattern on $3$ or $4$ men, is
%   \\
%     \begin{tabular}{c|cccccc}
%       1 & a & b & c & d \\
%       2 & a & b & c & d \\
%       3 & a & c & b & d \\
%       4 & b & a & d & c
%     \end{tabular}
% \end{corollary}
% \begin{proof}
% 
%   By section~\ref{sec:twoNis4Proofs}, the second and fifth preference
%   list above are not OSP.
%   % The third and fourth are actually equivalent up to relabeling,
%   So we have that the only cyclic but OSP preference lists when
%   $n=4$ are, up to relabeling,
%   \\
%     \begin{tabular}{c|cccccc}
%       1 & d & a & b & c \\
%       2 & d & a & b & c \\
%       3 & d & a & c & b \\
%       4 & d & b & a & c
%     \end{tabular}
%     \qquad
%     \begin{tabular}{c|cccccc}
%       1 & d & a & b & c \\
%       2 & d & a & b & c \\
%       3 & a & d & c & b \\
%       4 & d & b & a & c
%     \end{tabular}
%     \qquad
%     \begin{tabular}{c|cccccc}
%       1 & a & b & c & d \\
%       2 & a & b & c & d \\
%       3 & a & c & b & d \\
%       4 & b & a & d & c
%     \end{tabular}
%     \qquad
%     \begin{tabular}{c|cccccc}
%       1 & a & b & c & d \\
%       2 & a & b & c & d \\
%       3 & a & c & b & d \\
%       4 & b & a & c & d
%     \end{tabular}
%   \\
%   The first and the last do not satisfy the hypothesis that all men mix
%   together a bit. The second and the third are equivalent up to relabeling.
% \end{proof}

Out of the six priority sets in the above lemma, four are limited cyclic
(and thus OSP). The other two, namely
\\
\begin{tabular}{cccc}
  \begin{tabular}{c|cccccc}
    1 & a & {\bf d} & b & c \\
    2 & {\bf d} & a & b & c \\
    3 & {\bf d} & a & c & b \\
    4 & {\bf d} & b & a & c
  \end{tabular}
  & and
  &
  $q = $
  \begin{tabular}{c|cccccc}
    1 & a & b & c & {\bf d} \\
    2 & a & b &{\bf  d} & c \\
    3 & a & c & b & {\bf d} \\
    4 & b & a & c & {\bf d}
  \end{tabular}
\end{tabular}
\\
are equivalent up to relabeling (to see this, relabel the applicants in the
first case as follows: $a \mapsto b, b \mapsto c, c\mapsto d, d\mapsto a$).
We show that this priority set is not OSP in \autoref{sec:nIs4CasesProof}.

Recall that \autoref{fig:AllPrimitiveNonOsp} consists exactly of those
priority sets from \autoref{fig:AllPrimitiveNonOspThreeByThree},
along with $q$ as above.

% \autoref{lem:4by4Cases} highlights the additional cases we need to check
% for $4$ applicants and $4$ positions.
% It turns out that there is only one new case up to relabeling.
% We show that this case is non-OSP in \autoref{sec:nIs4CasesProof}.
% Using this fact and \autoref{lemma:superNotTransitive},
% we achieve the following result, which is our crucial building block to the
% general proof.

% \ctnote{rephrase next stuff}
\begin{corollary}
  \label{lem:4by4}
  The only set of priorities of $4$ positions
  over $4$ applicants which are cyclic, yet do not contain any of the
  priority lists in \autoref{fig:AllPrimitiveNonOsp}, are equal to one of
  the following (up to relabeling):
  \\
  \begin{tabular}{ccc}
    \begin{tabular}{c|cccccc}
      1 & {\bf d} & a & b & c \\
      2 & {\bf d} & a & b & c \\
      3 & {\bf d} & a & c & b \\
      4 & {\bf d} & b & a & c
    \end{tabular}
    &
    \qquad
    \begin{tabular}{c|cccccc}
      1 & a & b & c & {\bf d} \\
      2 & a & b & c & {\bf d} \\
      3 & a & c & b & {\bf d} \\
      4 & b & a & {\bf d} & c
    \end{tabular}
    &
    \qquad
    \begin{tabular}{c|cccccc}
      1 & a & b & c & {\bf d} \\
      2 & a & b & c & {\bf d} \\
      3 & a & c & b & {\bf d} \\
      4 & b & a & c & {\bf d}
    \end{tabular}
  \end{tabular}.
\end{corollary}
\begin{proof}
  % All of the listed cases are limited cyclic,
  % and thus OSP by \autoref{thrm:Construction}.
  By \autoref{lem:4by4Cases} and the fact that $q$ (as defined above)
  is in \autoref{fig:AllPrimitiveNonOsp},
  there are only $4$ possible cases.
  Three of these are given, and the final is
  \\
  \begin{tabular}{c|cccccc}
    1 & d & a & b & c \\
    2 & d & a & b & c \\
    3 & a & d & c & b \\
    4 & d & b & a & c
  \end{tabular}
  \\
  which is equivalent up to relabeling to the middle priority set in the
  statement of the corollary.
  % \\
  % \begin{tabular}{cccc}
  %   \begin{tabular}{c|cccccc}
  %     1 & a & d & b & c \\
  %     2 & d & a & b & c \\
  %     3 & d & a & c & b \\
  %     4 & d & b & a & c
  %   \end{tabular}
  %   &
  %   \begin{tabular}{c|cccccc}
  %     1 & d & a & b & c \\
  %     2 & d & a & b & c \\
  %     3 & a & d & c & b \\
  %     4 & d & b & a & c
  %   \end{tabular}
  %   &
  %   \begin{tabular}{c|cccccc}
  %     1 & a & b & c & d \\
  %     2 & a & b & d & c \\
  %     3 & a & c & b & d \\
  %     4 & b & a & c & d
  %   \end{tabular}
  % \end{tabular}.
  % \\
  % The second case above is equivalent (up to relabeling) to the second case
  % in the statement of the corollary. 
  % The first and the fourth case are also equivalent up to relabeling 
  % (to see this, relabel the applicants in the first case as follows:
  % $a \mapsto b, b \mapsto c, c\mapsto d, d\mapsto a$).
  % We show in \autoref{sec:nIs4CasesProof} that these are not OSP,
  % which finishes the proof.
\end{proof}

\subsection{The case with four applicants and positions}
\label{sec:nIs4CasesProof}

We now provide the single needed direct proof that a set of 
priorities with $4$ applicants and $4$ positions is not OSP.
This will complete the proof that all priority sets displayed in
\autoref{fig:AllPrimitiveNonOsp} are non-OSP.

\begin{claim}
  \label{claim:nIsForProof}
  Priorities
  \begin{align*}
    1: &\ a \succ b \succ c \succ d  \\
    2: &\ a \succ b \succ d \succ c  \\
    3: &\ a \succ c \succ b \succ d  \\
    4: &\ b \succ a \succ c \succ d 
  \end{align*}
  are not OSP.
\end{claim}
\begin{proof}

We apply \autoref{lem:keyTechnical} to the
the subdomain in which applicants may have the following
preferences:
\begin{align*}
  p_a^1 = 4,2 && p_b^1 = 3,1 && p_c^1 = 2,3 && p_d^1 = 1,2 \\
  p_a^2 = 4,3 && p_b^2 = 3,4 && p_c^2 = 3,1 && p_d^2 = 2,1 \\
\end{align*}
For simplicity, we do not specify the full preference lists here, only
those positions which are relevant to the argument (filling in the
remaining positions in any way will leave the argument unchanged).

\begin{itemize}
  \item $a$ acts at the root node.

    ``truth getting a bad result'': $p_a^1, p_b^2, p_c^2, p_d^1$:
    \begin{tabular}{c|c|c|c}
      1 & d   &   &     \\
      2 &     &   & a   \\
      3 & b c &   &     \\
      4 & a   & b &    
    \end{tabular}.

    ``lie getting a good result'':  $p_a^2, p_b^2, p_c^1, p_d^1$:
    \begin{tabular}{c|c}
      1 & d    \\
      2 & c    \\
      3 & b    \\
      4 & a   
    \end{tabular}.

  \item $b$ acts at the root node.

    ``truth getting a bad result'': $p_a^1, p_b^1, p_c^2, p_d^2 $:
    \begin{tabular}{c|c|c}
      1 &     & b   \\
      2 & d   &     \\
      3 & b c &     \\
      4 & a   &    
    \end{tabular}.

    ``lie getting a good result'':  $p_a^1, p_b^2, p_c^1, p_d^2 $:
    \begin{tabular}{c|c|c}
      1 &     & d   \\
      2 & d c &     \\
      3 & b   &     \\
      4 & a   &    
    \end{tabular}.

  \item $c$ acts at the root node.

    ``truth getting a bad result'': $p_a^2, p_b^2, p_c^2, p_d^2$:
    \begin{tabular}{c|c|c|c|c}
      1 &     &   &   & c   \\
      2 & d   &   &   &     \\
      3 & b c &   & a &     \\
      4 & a   & b &   &    
    \end{tabular}.

    ``lie getting a good result'': $p_a^2, p_b^1, p_c^1, p_d^2$:
    \begin{tabular}{c|c|c|c}
      1 &     &   & b    \\
      2 & d c &   &      \\
      3 & b   & c &      \\
      4 & a   &   &     
    \end{tabular}.

  \item $d$ acts at the root node.

    ``truth getting a bad result'': $p_a^1, p_b^1, p_c^2, p_d^1$:
    \begin{tabular}{c|c|c|c}
      1 & d   & b &    \\
      2 &     &   & d  \\
      3 & b c &   &    \\
      4 & a   &   &  
    \end{tabular}.

    ``lie getting a good result'': $p_a^1, p_b^2, p_c^2, p_d^2$:
    \begin{tabular}{c|c|c|c|c}
      1 &     &   &   & d \\
      2 & d   &   & a &   \\
      3 & b c &   &   &   \\
      4 & a   & b &   &
    \end{tabular}.

\end{itemize}

\end{proof}

\subsection{Proof of main result}
\label{sec:nAtLeastFourFinalProof}

We are now ready to complete our proof.
Recall that \autoref{fig:AllPrimitiveNonOsp} contains exactly those
preference sets of \autoref{fig:AllPrimitiveNonOspThreeByThree} along with 
that of \autoref{sec:nIs4CasesProof}.
Intuitively, the proof starts from \autoref{lem:3by4},
then uses the fact that the priority set of \autoref{sec:nIs4CasesProof}
is the \emph{only} nontrivial way to add \emph{one} more applicants to a cyclic and OSP
priority list. The new applicant must be placed in such a particular way
that the priority list of \autoref{sec:nIs4CasesProof} is also 
the only way to add \emph{more} applicants.
% to show that the only way to add \emph{more} applicants is % also analogous.
The crux of the formal proof is to show that the two positions $u$ and $v$
at which the priority list ``deviates from the norm'' must be consistent
with each further applicant added.
% Informally, the proof amounts to showing that 

\begin{theorem*}[4.2]
  Consider any priority set $q$.
  If no restriction of $q$ is equal to any of the priority sets of
  \autoref{fig:AllPrimitiveNonOsp}, then $q$ is limited cyclic.
\end{theorem*}
% \begin{theorem}
%   \label{thrm:MainImposibility}
%   Let $q$ be any set of priority lists
%   of $n$ positions over $n$ applicants.
%   If $\DA^q$ is OSP implementable,
%   then $q$ is limited cyclic.
% \end{theorem}
\begin{proof}
  % We prove the following more general claim:
  % For any value of $m$, if $q$ is a set of priorities of $n$ positions over
  % $m$ applicants, then $q$ is limited cyclic.
  % The proof procedes by induction on $m$.

  % \ctnote{here}
  Let $q$ satisfy the hypothesis of the theorem.
  Fix any position $i^*$ and construct an ordered partition
  $S_1,S_2,\ldots,S_L$ of applicants as follows:
  take the priority list $\succ_{i^*}$ and iteratively group together any
  applicants $a, b$ with $a \succ_{i^*} b$ such that there exists a 
  position $j$ where $b \succ_j a$.
  That is, $\{S_1, \ldots, S_L\}$ is the finest possible partition of
  applicants such that, 1) for all $a\in S_k, b\in S_\ell$ with 
  $k < \ell$, we have $a \succ_{i^*} b$,
  and 2) if $a \succ_{i^*} b$ and $b \succ_j a$ for some
  $j \ne i^*$, then $a, b$ are in the same set $S_i$.
  % \footnote{ \ctnote{SOMETHING like this should eventually be around}
  %   Note that
  %   the trivial partition $\{S_1\}$ where all applicants are in $S_1$
  %   satisfies requirements 1) and 2).
  %   Moreover, 
  %   This ordered partition is uniquely determined, because
  % }.
  % and we order 
  % $S_1,\ldots,S_L$ according to the priority list $\succ_{i^*}$.
  % By definition,   % \footnote{
  %   Such an ordering must exist. To see this, let $S_k > S_\ell$
  %   denote the fact that there exists an $a \in S_k$ and $b \in S_\ell$
  %   with $a \succ_{i^*} b$. Then, for each $S_k$ and $S_\ell$, one of
  %   $S_k > S_\ell$ or $S_\ell > S_k$ hold.
  %   If both $S_k > S_\ell$ and $S_\ell > S_k$ hold, 
  % }.
  %  if we ever had $b \succ_j a$ for $j\ne i^*$, then we would have
  % $a$ and $b$ in the same $S_i$.
  Observe that for any $a \in S_k, b \in S_\ell$ with $k < \ell$,
  we have $a \succ_j b$ for \emph{all} positions $j$ (this follows because
  $a \succ_{i^*} b$ by the way we construct the partition, so if we ever
  had $b \succ_j a$ for some other position $j$, then $a, b$ should be in
  the same set $S_i$).
  % To prove the theorem, it suffices to show that for all 
  Thus, to prove our theorem, we just need to show that for each $S_i$ with
  $|S_i| \ge 3$, the priorities restricted to $S_i$ are
  two-adjacent-alternating as in \autoref{def:limitedCyclic}.

  To prove this, fix some $i$ and let $|S_i|=K\ge 3$.
  We use induction on $K$.
  Let $q'$ be $q$, with each priority list restricted to applicants
  in $S_i$. 
  % By \autoref{lem:RestrictedPriorities}, no restriction of $q'$ 
  % (to some subset of applicants and positions) can be non-OSP.
  By \autoref{prop:AcyclicEquiv}, $q'$ is cyclic.
  Thus, if $K=3$, then \autoref{lem:3by4} proves the base case.

  Otherwise, let $|S_i|=K\ge 4$. 
  % Assume by induction that for any $S_i'$ with $|S_i'|\le K-1$
  % such that no restriction of $q$ is equal to any of the priority sets of
  % \autoref{fig:AllPrimitiveNonOsp},
  Assume by induction that for any such\footnote{
    Formally, the inductive hypothesis is that if $q'$ is any priority set
    over $K$ applicants such that 1) no restriction of $q'$ is equal to any
    of the priority sets in \autoref{fig:AllPrimitiveNonOsp}, and
    2) the partition constructed on $q'$ as in the first paragraph of 
    this theorem contains only a single set, 
    then $q'$ is two-adjacent-alternating.
  } $S_i'$ with $|S_i'|\le K-1$, the priorities restricted to $S_i'$ are 
  two-adjacent-alternating.
  Let $q'$ be $q$, with each priority list restricted to applicants
  in $S_i$. 
  By \autoref{lemma:allFavorites}, there can be at most two applicants
  $a_1, a_2$ such that some positions gives $a_1, a_2$ top priority in $q'$.
  Moreover, there must exist some $a_3$ such that $q'$ restricted to
  applicants $\{a_1,a_2,a_3\}$ is cyclic (by \autoref{prop:AcyclicEquiv}, 
  if this were not the case, $S_i$ would equal $\{a_1,a_2\}$ ).
  By \autoref{lem:3by4}, each priority list in $q'$, when restricted to
  $\{a_1,a_2,a_3\}$, must be equal to $x_0, u_0, v_0$, where 
  $u_0$ and $v_0$ each appear exactly once (at positions $u$
  and $v$, respectively), and we have
  \\
  \begin{tabular}{c|cccc}
    $x_0$ & $a_1$ & $a_2$ & $a_3$ \\
    $u_0$ & $a_1$ & $a_3$ & $a_2$ \\
    $v_0$ & $a_2$ & $a_1$ & $a_3$
  \end{tabular}

  As $|S_i|\ge 4$, there must exist another applicant $a_4$ such that $a_4$
  appears before one of $a_1, a_2,$ or $a_3$ on some position's priority list.
  Let $y$ be any position which ranks $a_4$ above at least one of $a_1,a_2,$
  or $a_3$.
  Consider any restriction $q''$ of $q'$ to any four positions including 
  $y, u, v$.
  % For any position $x'$ whose priority list is of the form $x$ when restricted
  % to $a,b,c$, consider the priority lists of $x', y, u, v$ restricted to
  % applicants $a,b,c,d$.
  % Consider the subset of positions which include $u, v,$ the position
  % ranking $d$ highly, and any other position $x'$.
  This forms a cyclic, OSP priority set of four positions for four
  applicants, and thus \autoref{lem:4by4} applies.
  The only case from \autoref{lem:4by4} in which $d$ is sometimes (but not
  always) ranked above one of $a,b,c$ is the middle case.
  That is, if we let
  \\
  \begin{tabular}{c|cccc}
    $x_1$ & $a_1$ & $a_2$ & $a_3$ & $a_4$ \\
    $u_1$ & $a_1$ & $a_3$ & $a_2$ & $a_4$ \\
    $v_1$ & $a_2$ & $a_1$ & $a_4$ & $a_3$
  \end{tabular},
  \\
  then $q''$ restricted to $\{a_1,a_2,a_3,a_4\}$
  must contain $x_1$ repeated twice, and each of $u_1, v_1$ exactly once.
  In particular, we must have $y=v$, and $u$ must have priority list 
  $u_1$ in $q''$ and $v$ must have $v_1$.
  This argument applies when we consider any positions $i, i'$
  together with $u, v$, and thus we can see that \emph{all}
  positions other than $u$ and $v$
  must have priority list $x_1$ when restricted to $\{a_1,a_2,a_3,a_4\}$.
  This argument also applies to any applicant $a_4$ who is ever ranked
  above any of $a_1, a_2,$ or $a_3$, so
  in particular note that $a_1$ has higher priority at every position
  applicant in $q'$ other than $a_2$ (at position $v$).
  % As this argument applies for any $x'$, we know that $q'$ is
  % two-adjacent-alternating when restricted to $\{a,b,c,d\}$.

  % In particular, $v$ is the unique position where $d$ is ranked above $c$.

  Now, consider removing $a_1$ from each position's priority list.
  What remains is still a cyclic set of priority lists with at least
  three applicants. Because $a_1$ was only ever ranked below $a_2$, and no
  other applicants, the applicants in $S_i\setminus \{a_1\}$ 
  cannot be partitioned further and still satisfying the assumptions on the
  partition $(S_i)_i$.
  Thus, by induction, priorities are
  two-adjacent-alternating when restricted to $S_i\setminus \{a\}$.
  We already know that, restricted to $\{a_2,a_3,a_4\}$, priorities $q'$ 
  are of the form
  \\
  \begin{tabular}{c|cccc}
    $x_2$ & $a_2$ & $a_3$ & $a_4$ \\
    $u_2$ & $a_3$ & $a_2$ & $a_4$ \\
    $v_2$ & $a_2$ & $a_4$ & $a_3$
  \end{tabular},
  \\
  where $u_2$ is the priority list of $u$ and $v_2$ is the priority list of $v$
  (and all other positions have priority lists $x_2$).
  Thus, the only way for priorities to be two-adjacent-alternating
  when restricted to $S_i \setminus\{a_1\}$ is if we have
  \\
  \begin{tabular}{c|ccccccccc}
       $x_3$ & $a_2$ & $a_3$ & $a_4$ & $a_5$ & $a_6$ & $a_7$ \ldots
    \\ $u_3$ & $a_3$ & $a_2$ & $a_5$ & $a_4$ & $a_7$ & $a_6$ \ldots
    \\ $v_3$ & $a_2$ & $a_4$ & $a_3$ & $a_6$ & $a_5$ & \ldots
  \end{tabular}.
  \\
  Moreover, it must be the case that
  $u$ has priority $u_3$, $v$ has priority $v_3$, and all others have
  priority $x_3$.
  Thus, with $a_1$ considered as well, priority lists must be of the form
  \\
  \begin{tabular}{c|ccccccccc}
       $x_4$ & $a_1$ & $a_2$ & $a_3$ & $a_4$ & $a_5$ & $a_6$ & $a_7$ \ldots
    \\ $u_4$ & $a_1$ & $a_3$ & $a_2$ & $a_5$ & $a_4$ & $a_7$ & $a_6$ \ldots
    \\ $v_4$ & $a_2$ & $a_1$ & $a_4$ & $a_3$ & $a_6$ & $a_5$ & \ldots
  \end{tabular}.
  \\
  That is, $q'$ is two-adjacent-alternating.
  % Moreover, because we have
  % \\
  % \begin{tabular}{c|cccc}
  %   $x_2$ & $b$ & $c$ & $d$ \\
  %   $u_2$ & $c$ & $b$ & $d$ \\
  %   $v_2$ & $b$ & $d$ & $c$
  % \end{tabular},
  % \\
  % the only possible way for priorities to be two-adjacent-alternating on
  % $S_1\setminus \{a\}$ is if $u, v$ are the alternating two.
  % Thus, priorities are two-adjacent-alternating with $a$ added back in as
  % well.

  By induction, this proves that $q'$ is two-adjacent alternating, and thus
  that $q$ is limited cyclic.
\end{proof}

\section{Extension to unbalanced matching markets}
\label{sec:Extensions}

% Extensions appendix

In this appendix, we briefly discuss the possibility of extending our
classification to a matching environment with a different number of
applicants and positions, i.e. ``unbalanced'' environments.
We consider the environment with ``outside options'', i.e. in
which applicants may mark positions as unacceptable (and chose to go
unmatched instead of being matched to unacceptable positions)\footnote{
  It is possible to consider an environment with imbalance, yet no outside
  option. In this environment, applicants will always rank every position
  to try and avoid being unmatched.
  Such an environment leads to somewhat pathological matching rules.
  For example,
  consider two positions $1,2$ and three applicants $a,b,c$, with priorities
  $1: a \succ c \succ b$ and $2: b \succ c \succ a$.
  In the environment with no outside options, there is a unique stable
  matching $\{(1,a),(2,b)\}$
  with these priorities \emph{regardless of the preferences of the
  three applicants}.
  On a more technical side, our crucial \autoref{lem:RestrictedPriorities}
  does not hold in unbalanced environments without outside options.
  Moreover, it is our view that if the mechanism designer intends to leave
  certain agents unmatched, they should model the possibility that the
  agent actually wishes to go unmatched.
  % \footnote{
  %   To see this, observe that matching $\{(1,a),(2,b)\}$ will be achieved
  %   just by considering the proposals made by $a$ and $b$,
  %   except in the case where preferences are $a : 2 \succ 1$ and 
  %   $b: 1 \succ 2$. Next, in this case observe that regardless of 
  %   $c$'s preferences, one of $a$ or $b$ will be displaced the matching will
  %   also result in $\{ (1,a),(2,b) \}$.
  % }. 
  % Due to examples like this, achieving a classification in this model may
  % be more delicate than in the main environment of this paper.
  % However, we do not view this model as particularly fundamental,
  % because it posits that agents have $-\infty$ utility for going unmatched,
  % yet it actual does leave agents unmatched in many outcomes.
  % This is a counter-intuitive axiom: if the mechanism designer intends to
  % genuinely leave some agent unmatched, then she should consider the
  % possibility that agents will \emph{want} to be unmatched.
}.

It is possible for non-limited cyclic priorities to be OSP in this
environment. Specifically, consider the following priorities, for which
\autoref{fig:ThreeTraderInconsistencies} describes an OSP mechanism:
\\
$q =$
\begin{tabular}{c|cccccc}
  1 & a & b & (d & c) \\
  2 & a & (c & b) & d \\
  3 & (b & a) & c & d
\end{tabular}
\\
Interestingly, in this OSP mechanism four agents can be active at a time
(in the rightmost node of the first row, where $c$ acts).
However, the mechanism still
interacts with each agent at most twice.

% Interestingly, this particular set of preferences is the \emph{only}
% example of an OSP but not limited cyclic preference set for this
% environment.
This counterexample can be extended only in very limited ways
-- all counterexamples must have exactly three positions.
To see this, first observe that \autoref{lem:RestrictedPriorities} can be
adapted to unbalanced environments with outside options (by having
applicants outside the restriction submit empty preference lists).
Whenever there are four or more positions and three or more applicants, 
\autoref{lem:3by4} and/or the proof of \autoref{thrm:patternMatchingMainThrm}
apply and show that OSP priorities must be limited cyclic.

\begin{figure}
\tikzset{
  solid node/.style={circle,draw,inner sep=1.5,fill=black},
  decision node/.style={regular polygon, regular polygon sides=6,draw,inner sep=1.5},
  hollow node/.style={circle,draw,inner sep=1.5},
  square node/.style={draw,inner sep=2.5},
  tri node/.style={regular polygon, regular polygon sides=3,draw,inner sep=1.5},
  % node gap={xshift=70} 
}
\begin{center}
\begin{tikzpicture}[scale=1.0,font=\footnotesize]
  \tikzstyle{level 1}=[level distance=15mm,sibling distance=12mm]
  \tikzstyle{level 2}=[level distance=15mm,sibling distance=5mm]
  \tikzstyle{level 3}=[level distance=15mm,sibling distance=7mm]
  \node(a1)[decision node,label=above:{}]{(a)}
    child{node(a11)[square node, label=below:{
        % $\TwoTrader^{2;3}_{b,c}$
      }]{}
      edge from parent node[left,yshift=-5,xshift=-2]{$\mathtt{C}(1)$}
    }
    child{node(a12)[square node, label=below:{
        % $\TwoTrader^{1;3}_{b,c}$
      }]{}
      edge from parent node[left,yshift=-5,xshift=2]{$\mathtt{C}(2)$}
    }
    child{node(a13)[tri node, label=below:{
        % $\TwoTrader^{1;3}_{b,c}$
      }]{}
      edge from parent node[left,yshift=-5,xshift=5]{$\mathtt{C}(\emptyset)$}
    };
  \node(b1)[decision node,label=above:{}, xshift=100]{(b)}
    child{node(b11)[square node, label=below:{
        % \begin{tabular}{c}
        %   $\{ (a,3),(b,1),$ \\
        %   $(c,2) \}$
        % \end{tabular}
        }]{}
      edge from parent node[left,yshift=-5,xshift=-2]{$\mathtt{C}(1)$}
    }
    child{node(b12)[square node, label=below:{
        % $\SerialDict^{1,2}_{a,c}$
        % \begin{tabular}{c}
        %   $a$ acts,\\ then $c$.
        % \end{tabular}
      }]{}
      edge from parent node[left,yshift=-5,xshift=2]{$\mathtt{C}(3)$}
    }
    child{node(b13)[square node, label=below:{
      }]{}
      edge from parent node[left,yshift=-5,xshift=5]{$\mathtt{C}(\emptyset)$}
    };
  \node(d1)[decision node,label=above:{$\{1\}$ is possible}, xshift=180]{(d)}
    child{node(c11)[decision node, label=below:{
          % \begin{tabular}{c}
          % $\{(a,3),(b,2),(c,1)\}$
          % \end{tabular} 
      }]{(i)}
      edge from parent node[left,xshift=-3]{$\mathtt{C}(\emptyset)$}
    };
  \node(c1)[decision node,label=above:{$\{2\}$ is possible}, xshift=260]{(c)}
    child{node(c11)[hollow node, label=below:{
          % \begin{tabular}{c}
          % $\{(a,3),(b,1),$ \\
          % $(c,2)\}$
          % \end{tabular} 
        }]{}
      edge from parent node[left,xshift=-3]{$\mathtt{C}(\emptyset)$}
    }
    child{node(b23)[decision node, label=below:{
        % $\SerialDict^{1,2}_{a,c}$
        % \begin{tabular}{c}
        %   $a$ acts
        % \end{tabular}
      }]{(ii)}
      edge from parent node[right,xshift=3]{$2\ldots$}
    };

  \draw[->] (a1) -- (b1) 
    node[midway, above]{$3\dots$};
  \draw[->] (b1) -- (d1) 
    node[midway, above]{$2\dots$};
  \draw[->] (d1) -- (c1) 
    node[midway, above]{$1\dots$};
  % \draw[dashed,rounded corners=10]($(1) + (-.2,.25)$)rectangle($(2) +(.2,-.25)$);
  % \node at ($(1)!.5!(2)$) {Bob: $I_B^1$};
\end{tikzpicture}

\begin{tikzpicture}[scale=1.0,font=\footnotesize]
  \tikzstyle{level 1}=[level distance=15mm,sibling distance=12mm]
  \tikzstyle{level 2}=[level distance=15mm,sibling distance=5mm]
  \tikzstyle{level 3}=[level distance=15mm,sibling distance=7mm]
  \node(ii)[decision node,label=right:{:}, xshift=30]{(ii)};
  \node(b2)[decision node,label=above:{$\{1,3\}$ is possible}, xshift=70]{(b)}
    child{node(b21)[hollow node, label=below:{
        % $\TwoTrader^{2;3}_{b,c}$
      }]{}
      edge from parent node[left,yshift=-5,xshift=-2]{$\mathtt{C}(1)$}
    }
    child{node(b22)[hollow node, label=below:{
        % $\TwoTrader^{1;3}_{b,c}$
      }]{}
      edge from parent node[left,yshift=-5,xshift=4]{$\mathtt{C}(\emptyset)$}
    };
  \node(a2)[decision node,label=above:{}, xshift=150]{(a)}
    child{node(a21)[hollow node, label=below:{
        }]{}
      edge from parent node[left,yshift=-5,xshift=-2]{$\mathtt{C}(1)$}
    }
    child{node(a22)[hollow node, label=below:{
        }]{}
      edge from parent node[left,yshift=-5,xshift=2]{$\mathtt{C}(2)$}
    }
    child{node(a23)[hollow node, label=below:{
        }]{}
      edge from parent node[left,yshift=-5,xshift=4]{$\mathtt{C}(\emptyset)$}
    };

  \draw[->] (b2) -- (a2) 
  node[midway, above]{$\mathtt{C}(3)$};
\end{tikzpicture}
\end{center}
\caption{
  An OSP set of priorities which is not limited cyclic (in an unbalanced
  matching environment).
  In the square subnodes, remaining priorities are acyclic;
  in the triangular subnodes, the remaining priorities are limited cyclic;
  and in the circular subnodes, the match is fully determined.
  At node $(i)$, applicant $d$ drops out, and the mechanism proceeds in the
  same way as in the balanced case in \autoref{fig:ThreeTraderAshlagiG}.
  The mechanism below node $(ii)$ is continued in the lower part of the figure.
}
\label{fig:ThreeTraderInconsistencies}
\end{figure}
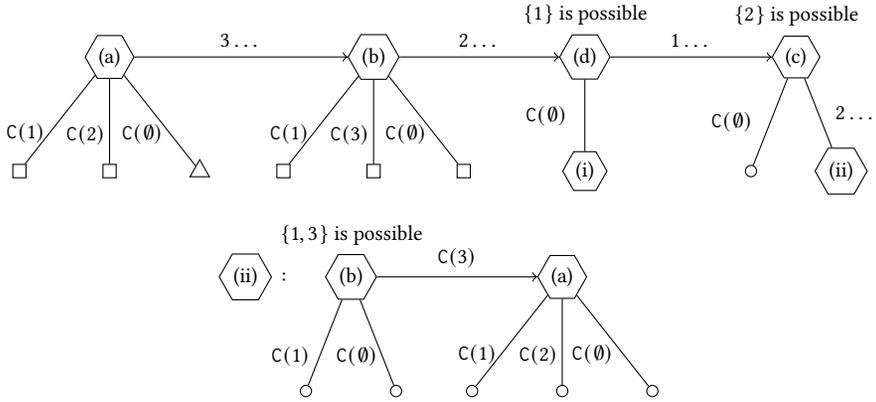

\end{document}